\documentclass[10pt,a4paper]{amsart}
\usepackage{latexsym,amssymb,amsmath,amsthm,amsfonts,enumerate,verbatim,xspace, exscale}
\usepackage{graphicx}
\usepackage{color,amsbsy,textcomp}
\usepackage{enumerate}

\usepackage{float}

\usepackage{graphicx}
\usepackage{rotating}

\usepackage{hyperref}
\hypersetup{
    colorlinks=true,
    linkcolor=blue,
    filecolor=magenta,      
    urlcolor=cyan,
}

\theoremstyle{plain}
\newtheorem{theorem}{Theorem}[section]
\newtheorem{lemma}[theorem]{Lemma}

\theoremstyle{definition}

\newtheorem{remark}[theorem]{Remark}
\newtheorem{ass}{Assumption}[section]

\newcommand{\up}{\upshape}

\def\vv<#1>{\langle#1\rangle}

\newcommand{\Om}{\Omega}

\newcommand{\todo}[1]{\vspace{5 mm}\par \noindent
\marginpar{\textsc{ToDo}} \framebox{\begin{minipage}[c]{0.95
\textwidth}\raggedright \tt #1 \end{minipage}}\vspace{5 mm}\par}

\def\vv<#1>{\langle#1\rangle}

\newcommand{\revise}[1]{#1} 

\usepackage[normalem]{ulem}
\usepackage{xcolor}
\usepackage{siunitx}

\newcommand\deleteF{\bgroup\markoverwith{\textcolor{blue}{\rule[0.8ex]{2pt}{0.9pt}}}\ULon}
\newcommand\deleteS{\bgroup\markoverwith{\textcolor{green}{\rule[0.8ex]{2pt}{0.9pt}}}\ULon}

\title[Estimation of future discretionary benefits]%
{
Estimation of future discretionary benefits in traditional life insurance
}

\author[Florian Gach \; \& \; Simon Hochgerner]{Florian Gach \qquad\&\qquad Simon Hochgerner}

\usepackage[foot]{amsaddr} 
\address{Austrian Financial Market Authority (FMA),
  Otto-Wagner Platz 5, A-1090 Vienna
} 
\email{florian.gach@fma.gv.at}
\email{simon.hochgerner@fma.gv.at}

\thanks{\emph{Disclaimer.} 
The opinions expressed in this article are those of the authors and do not necessarily reflect the official position of the Austrian Financial Market Authority. }


\usepackage{longtable}
\begin{document}

\begin{abstract}
In the context of life insurance with profit participation, the future discretionary benefits ($FDB$), which are a central item for Solvency~II reporting, are generally calculated by computationally expensive Monte Carlo algorithms.
We derive analytic formulas to estimate lower and upper bounds for the $FDB$. This yields an estimation interval for the $FDB$, and the average of lower and upper bound is a simple estimator. These formulae are designed for real world applications, and we compare the results to publicly available reporting data.
\end{abstract}

\maketitle


\section{Introduction}
As of 1.\ 1.\ 2016 the European Union has implemented a new regulatory regime (Solvency~II \cite{L1, L2, L3templ}) which requires insurance companies operating in Member States  to assign market consistent values to their balance sheet items. This requirement concerns, both, assets and liabilities and is therefore a full balance sheet approach. 

Market consistent valuation in the context of life insurance with profit participation has been a developing subject over the past two decades. Early works in this context include \cite{SS04, OBrien09, Delong,  W16}. These works all highlight the interdependencies that exist between the insurer's asset portfolio and the policyholder's expected payoff. Indeed, the defining feature of with profit contracts is that the policyholder's benefit is a sum of, firstly, a guaranteed part depending, in particular, on the guaranteed technical interest rate, and, secondly, a bonus benefit depending, in particular, on the performance of the company's asset portfolio. 

To obtain a more realistic model, the performance of the asset portfolio should be measured as a return on balance sheet items and, accordingly, there are recent works which incorporate a strongly simplified, or stylized, version of a balance sheet projection (\cite{Gerstner08, Gerstner09, Bacin21, Engsner_etal17, FN21}). 

However,  to have a \revise{more accurate} representation of the company's financial income, which is to be shared with policyholders, the return on assets should be derived according to local generally accepted accounting principles (local GAAP). This necessity is explained in detail in \cite{Dorobantu_etal20, D_etal17}. Moreover, financial revenue can be controlled to a certain extent by management actions by realizing unrealized gains of individual assets, i.e.\ the difference of market and local GAAP book value (\cite{Dorobantu_etal20}). Together with further management actions (e.g., strategic asset allocation, reinvestment strategy,  profit sharing and profit declaration) this setup leads, in practice, to computationally expensive Monte Carlo algorithms in order to obtain  realistic calculations of market consistent values of liabilities. See \cite{VELP17, Dorobantu_etal20, D_etal17, Albrecher_etal18}. 

Under Solvency~II the market consistent value of liabilities is defined as a sum of a best estimate and a risk margin (\cite[Article~77]{L1}). For the purpose of life insurance with profit participation, the risk margin is generally several orders of magnitudes smaller than the best estimate (\cite{eiopa_is}), as shown in 
Table~\ref{table:RM}.
\begin{table}[H]
\[
\begin{matrix}
		& 2016 &	 2017 &	 2018 &	 2019 	& 2020  \\
\textup{Best Estimate} &	 \num{4 759 241.59} 	& \num{4 893 285.81} &	\num{4 812 203.51}	& \num{5 331 803.61} &	\num{5 459 642.07} \\ 
\textup{Risk margin} &	 \num{79 574.17} &	 \num{77 532.75} &	\num{76 709.35} &	\num{83 692.55} & 	\num{89 736.47} \\
 &		1.67\,\%	& 1.58\,\% &	1.59\,\%	&1.57\,\%	&1.64\,\%
\end{matrix}
\]\vspace{1mm}
    \caption{Risk margin vs.\ best estimate. Aggregate of solo companies in EEA from 2016 to 2020 for life insurance (excluding health, index-linked and unit-linked). In million Euro and percent of best estimate. Source: Items R0670 and R0680 in `Balance sheet EEA' in \cite{eiopa_is};}
    \label{table:RM}
\end{table}

The best estimate of liabilities is defined as the expected value of discounted and probability weighted future policyholder and cost cash flows stemming from contracts which are active at valuation time. Thus new business is not considered and the expected value is to be taken with respect to a risk neutral measure. This definition of  \cite[Article~77]{L1} is clarified in \cite{L2}. 

Regarding life insurance with profit participation, the best estimate, denoted henceforth by $BE$, can be split into a sum of two parts: $BE = GB + FDB$; here $GB$ denotes the value of those cash flows which are guaranteed at valuation time while $FDB$ is the value of future discretionary benefits. Both, $FDB$ and the total value $BE$, have to be reported individually by insurance companies on a quarterly basis. The significance of this splitting is that the guaranteed benefits, $GB$, are calculated by methods which are close to classical actuarial computations. If one allows dynamic policyholder behavior such that, e.g., surrender depends on economic scenarios then a stochastic cash flow model is needed to calculate $GB$. Nevertheless, in comparison to the future discretionary benefits, the guaranteed benefits are more straightforward to calculate and validate. This is so because the $FDB$ depends directly on the company's surplus, and thus in particular on  management actions and financial revenue, while $GB$ depends on these quantities only indirectly via the part that is affected by  dynamic policyholder behavior. If a company chooses not to model dynamic policyholder behavior then $GB$ can be calculated by a deterministic actuarial cash flow model. 

Due to the above described complexity there exist no closed formula solutions for $FDB$ calculation in the context of life insurance with profit participation. An analytic formula for an estimate of a lower bound for $FDB$ has been derived in \cite{HG19}, and this has been applied to validate the reported $FDB$ of one major German  life insurance company for the reporting year 2017. The purpose of this paper is to improve and extend the approach of \cite{HG19} in the following ways:
\begin{enumerate}
    \item 
    In equation~\eqref{e:fdb-rep} we derive a representation of the $FDB$ which holds in a generic valuation framework. This representation cannot directly be applied to calculate the $FDB$ since it relies on quantities which are just as difficult to obtain as the $FDB$ itself. However, due to its generality, it can be used to validate the results of a given stochastic cash flow model, and it is the basis for the subsequently described estimations. 
    \item 
    Starting from equation~\eqref{e:fdb-rep} we derive in Section~\ref{sec:LBandUB} estimations of lower and upper bounds, $\widehat{LB}$ and $\widehat{UB}$, for the $FDB$. These estimated bounds are given by the analytic formulas \eqref{e:LB} and \eqref{e:UB}, and all the data needed to evaluate the latter are listed in Table~\ref{tbl:data_hat}.
    \item 
    The constituents of the lower and upper bounds, $\widehat{LB}$ and $\widehat{UB}$, are derived in an economically meaningful way such that there is an interpretation for each term and the  dependency on market conditions at valuation time is plausible. The latter point concerns specifically the dependence on the market view of interest rate volatility. This, apart from the new upper bound, is an improvement over the lower bound formula given in \cite{HG19}.  
    \item 
    Given $\widehat{LB}$ and $\widehat{UB}$ we obtain an estimator, $\widehat{FDB} = (\widehat{LB} + \widehat{UB})/2$, for $FDB$. This estimator is useful if the estimation error $\pm\epsilon$ with $\epsilon =  (\widehat{UB} - \widehat{LB})/2$ is sufficiently small. This can be measured against the market value, $MV_0$, of the company's portfolio at valuation time. 
    \item 
    In Section~\ref{sec:est}, we apply these formulas to publicly available data of a German life insurance company for the reporting years 2017, 2018 and 2019, and compare our results with their numerically calculated value of $FDB$. We find that $\epsilon < \revise{1.5}\,\%\, MV_0$ for all three reporting periods and, moreover, the difference $\delta = FDB-\widehat{FDB}$ satisfies $\delta < 1\,\%\, MV_0$ in all cases. Hence, in all three cases, the true value lies within the estimation interval and the estimation is remarkably accurate, as measured by $\delta$. See Table~\ref{table:Nr 1 }. 
    \item 
    Finally, in Section~\ref{sec:est} we also perform a sensitivity analysis with respect to the parameters that are used in calculations of $\widehat{LB}$ and $\widehat{UB}$ for the German life insurer, and find that the results are quite stable.  
\end{enumerate}

The main tools used to obtain  the representation~\eqref{e:fdb-rep} are the no-leakage principle of \cite{HG19} \revise{(see Remark~\ref{rem:no-leak})} and an evolution equation~\eqref{e:DBSFevol} for the statutory reserves of shared profits. Thus this derivation depends in a generic sense on the local GAAP framework but not on specific management rules. 

The derivation of the estimates, $\widehat{LB}$ and $\widehat{UB}$, depends on a number of assumptions. These are all listed and discussed in Section~\ref{sec:assump}. While the assumptions are tailored to the German and Austrian markets, their derivation is quite generic and is applicable whenever the company's revenue is given by well-defined local GAA Principles (e.g., as discussed for the French market in \cite{Dorobantu_etal20}).  The main observation in this context is that book values are expected to be more stable than market value movements. This conclusion is drawn from accounting principles for book values of assets (e.g. strict lower of cost or market price \cite{Dorobantu_etal20}), on one hand, and the use \revise{of} surplus funds as buffer accounts for statutory reserves \cite{Gerstner08}), on the other hand. 

The bounds \eqref{e:LB} and \eqref{e:UB} can be readily applied to real world data, as we  show in Section~\ref{sec:est}. Immediate and practical applications therefore include the following:
\begin{enumerate}
    \item 
    Internal validation: companies may use $\widehat{LB}$ and $\widehat{UB}$ to validate their $FDB$ calculations, and thus their valuation models. 
    \item
    External validation by parent companies: holdings may wish to validate the valuation models in their subsidiaries.
    \item 
    External validation by supervisors or auditors.
    \item
    Sensitivity analysis. 
    The estimator $\widehat{FDB}$ depends on economic quantities such as the prevailing interest rate and volatilities. These are at the same time important drivers for the full Monte Carlo best estimate calculation. Hence one obtains a closed estimation formula which allows for efficient sensitivity analysis without the necessity of stochastic computation.
\end{enumerate}
Clearly, the validation of the best estimate will be most effective when the control is paired with a statistical analysis of the second order assumptions leading to $GB$ and a verification that the contract specific features, which give rise to the guaranteed benefit cash flows, are correctly implemented. We do not view our estimation formula as a substitute for $FDB$ calculation as required by Solvency~II regulation.

\revise{
The derivation of the bounds \eqref{e:LB} and \eqref{e:UB} relies on a detailed investigation of cash-flows generated by local accounting principles. The notation that is used in this context is quite heavy and therefore collected in the Appendix at the end of the paper. 
}

\section{Framework} \label{sec: framework}
We consider a life insurance company selling traditional life insurance products. `Traditional' in this context means that the company's gross surplus is shared between policyholder, shareholder and tax office. Thus these contracts have a profit sharing feature. 

Further we assume that the life insurance company under consideration is subject to the Solvency II regulatory regime. 
 
Let us fix a yearly time grid, $t = 0, 1, \dots, T$, where $t=0$ corresponds to valuation time and $T$ may be large. For a time dependent quantity, $f_t$, we denote the increments by $\Delta f_t = f_t - f_{t-1}$\label{ref: Delta f}.

\subsection{Book value return on assets}
Let $C_t$ \label{ref: C_t} denote the amount of cash held by the company at time $t$. 
Let further $\mathcal{A}_t$\label{ref: A_t} denote the set of assets, other than the cash account, in the company's portfolio at time $t$. Each asset, $a\in \mathcal{A}_t$, has at each time step a book value, $BV_t^a$\label{ref: BV_t^a}, and a market value, $MV_t^a$\label{ref: MV_t^a}. The difference are the unrealized gains or losses, $UG_t^a = MV_t^a - BV_t^a$\label{ref: UG_t^a}. 
The total portfolio values are correspondingly
\[
 BV_t = \sum_{a\in\mathcal{A}_t}BV_t^a + C_t,\quad 
  MV_t = \sum_{a\in\mathcal{A}_t}MV_t^a + C_t,\quad 
 UG_t = MV_t - BV_t. 
\]\label{ref: BV_t} \label{ref: MV_t} \label{ref: UG_t}
In accordance with Solvency~II regulation, we assume that all market values are determined in an arbitrage free manner. 
Concretely, we assume an interest rate model specified on a filtered probability space $(\Om,\mathcal{F},(\mathcal{F}_t),\mathbb{Q)}$  such that $\mathbb{Q}$ is the risk neutral measure for the money market numeraire $B_t = \Pi_{j=0}^{t-1}(1+F_j)$\label{ref: B_t}, 
where $F_{t}$\label{ref: F_t} is the simple one year forward rate valid between $t$ and $t+1$. 
E.\,g., this setting is satisfied for the LIBOR market model with respect to the spot measure (\cite{BM06}).   
Further, all stochastic processes shall be adapted to $(\mathcal{F}_t)$, and all expected values, $E[\;\cdot\;]$, are with respect to $\mathbb{Q}$. 
The discount factor from $s$ to $t< s$ is given by \label{ref: D(t,s)} 
$
 D(t,s) 
 = \Pi_{j=t}^{s-1}(1+F_j)^{-1} 
 = B_t B_s^{-1}
$\revise{.}
Hence the value of a zero coupon bond paying one unit of currency at $s$ is given by $P(0,s) = E[B_s^{-1}]$, and more generally we have $P(t,s) = E[D(t,s)]$ for $t< s$\label{ref: P(t,s)}. 

Let $a\in\mathcal{A}_{t-1}$ and let $cf_s^a$ be the cash flow generated by $a$ at $s\ge t$. E.\,g., if $a$ is a bond, stock equity position or real estate investment, then $cf_t^a$\label{ref: cf_t^a} would correspond to a coupon or principle  payment, dividend yield or rental revenue, respectively. The cash flow, $cf_s^a$, goes to the cash account, $C_s$, and increases it accordingly. Cash flows can be influenced by market movements (affecting, e.g., dividend payments) or by management decisions. Indeed, the management may decide at any time, $s<T_a$, prior to the asset's maturity to sell the asset and realize the market value $MV_s^a$ as a cash flow. This step is called realization because it converts the unrealized gains $UG_s^a$ into a book value return. If an asset, $a$, is sold at $s<T_a$ then its book value is terminated, $BV_s^a=0$, and its market value is converted to a cash flow, $cf_s^a = (cf_s^a)' + MV_a^s$, where $(cf_s^a)'$ are those cash flows (e.g., coupon payments or dividend yield) that result from holding $a$ over the period $[s-1,s]$. 

The cash flow process, $cf_t^a$, is  an  $(\mathcal{F}_t)$-adapted process.

To avoid notational difficulties, we set $cf_s^a=0 = MV_s^a$ for $s>T_a$.  
At $t-1$ the forward rate $F_{t-1}$ is known,  i.e.\ it is $(\mathcal{F}_{t-1})$-adapted.
We have the relation
\begin{align}
    \label{e:MVa} 
    MV_{t-1}^a 
    &= 
    E\Big[\sum_{s=t}^{T}D(t-1,s)cf_s^a\Big| \mathcal{F}_{t-1}\Big] + E\Big[D(t-1,T)MV_T^a\Big| \mathcal{F}_{t-1}\Big] \\
    &= 
    \Big(1+F_{t-1}\Big)^{-1} \Big(E\Big[cf_t^a\Big|\mathcal{F}_{t-1}\Big] + E\Big[MV_t^a\Big| \mathcal{F}_{t-1}\Big]\Big) \nonumber
\end{align}
where the first line involves the final market value, $MV_T^a$, at the end of the projection which satisfies $MV_T^a = 0$ if $T_a\le T$.   

The book value return, $ROA_t^a$,\label{ref: ROA_t^a}  at $t$ of a single asset $a\in\mathcal{A}_{t-1}$ is
\begin{equation} \label{e:ROA of a single asset}
 ROA_t^a = cf_t^a + \Delta BV_t^a 
\end{equation}
and that of the portfolio is
\begin{equation}
    ROA_t = \sum_{a\in\mathcal{A}_{t-1}}ROA_t^a + F_{t-1}C_{t-1}
\end{equation} \label{ref: ROA_t}
since we have fixed yearly time steps and $F_{t-1}$ is the corresponding forward rate. 
Combining \eqref{e:ROA of a single asset} with \eqref{e:MVa} yields
the implication for the portfolio that
\begin{align}
    \label{e:ROA}
    ROA_t 
    &= 
    E[ROA_t|\mathcal{F}_{t-1}] + ROA_t - E[ROA_t|\mathcal{F}_{t-1}] \\
    \notag 
    &= 
    \sum_{a\in\mathcal{A}_{t-1}}\Big(
     E[cf_t^a|\mathcal{F}_{t-1}] 
     + E[\Delta BV_t^a|\mathcal{F}_{t-1}]
    \Big)
    + F_{t-1}C_{t-1}
    +  ROA_t - E[ROA_t|\mathcal{F}_{t-1}] \\
    \notag
    &=     
    \sum_{a\in\mathcal{A}_{t-1}}\Big(
     (1+F_{t-1})MV_{t-1}^a 
     - E[MV_t^a|\mathcal{F}_{t-1}]
     + E[\Delta BV_t^a|\mathcal{F}_{t-1}]
    \Big)
    + F_{t-1}C_{t-1}
    +  ROA_t - E[ROA_t|\mathcal{F}_{t-1}] 
    \\
    &= F_{t-1}BV_{t-1} 
      + F_{t-1}UG_{t-1} -  E[\revise{\Delta}UG_t|\mathcal{F}_{t-1}] 
      + ROA_t - E[ROA_t|\mathcal{F}_{t-1}]
\notag 
\end{align}
where we have used that $\sum_{a\in\mathcal{A}_{t-1}} E[UG_t^a|\mathcal{F}_{t-1}] = E[UG_t|\mathcal{F}_{t-1}]$. The significance of this equation is that it splits $ROA_t$ into three terms: the forward yield on the total book value, $F_{t-1}BV_{t-1}$; a term depending on the $\mathcal{F}_{t-1}$-prediction of return due to realizations of unrealized gains, $F_{t-1}UG_{t-1} -  E[\revise{\Delta} UG_t|\mathcal{F}_{t-1}]$; finally, the difference between return and $\mathcal{F}_{t-1}$-predicted return.  This interpretation is illustrated in the following toy example. 

\begin{remark}[Example]\label{rem:roaEx}
Fix an arbitrary $1<t<T$ and assume the company's portfolio at time $t-1$ consists of two identical bonds, denoted by $a_1$ and $a_2$, with maturity $t+1$, fixed coupon payment $KN/2$ and notional $N/2$, each, with $K\ge0$ and $N>0$. Suppose further that the company employs  the strict lower of cost or market value such that $BV_s = \textup{min}(N, MV_s)$ for $s=t-1$ and $s=t$. The cash flows generated by $a_i$ are coupon and notional payments or, if the management decides to sell $a_i$ at $s<t+1$, an additional payment of the prevailing market value at $s$ but no further (e.g., notional) payments after this time. We assume that the time $t-1$ predicted management rules are such that bonds are held to maturity, i.e., $E[cf_{t}^{a_i}|\mathcal{F}_{t-1}] = KN/2$ and $E[cf_{t+1}^{a_i}|\mathcal{F}_{t-1}] = (K+1)N/2$ for $i=1,2$.

The market value of the portfolio at $t$ is then $MV_t = (1+F_t)^{-1}(1+K)N$. At $t-1$, assume that   $K\ge F_{t-1}$ and  $E[(1+F_t)^{-1}|\mathcal{F}_{t-1}](1+K)\ge1$
such that 
\[
MV_{t-1} 
= (1+F_{t-1})^{-1}E[(1+F_t)^{-1}|\mathcal{F}_{t-1}](1+K)N + (1+F_{t-1})KN 
\ge N
\] 
whence $BV_{t-1}=\textup{min}(N,MV_{t-1}) = N$. The formula for $MV_{t-1}$ is, in fact, independent of management decisions at $t$ which reflects the fact that future investment strategies cannot affect the currently given market consistent value of a portfolio. Now we describe the constituents of equation~\eqref{e:ROA} under two management decisions at $t$: decision (1) is defined as the management's default strategy of not taking any action at $t$, while decision (2) shall mean that management decides to sell one of the two bonds at $t$, namely $a_1$. 
For example, this decision might depend on the observation of $F_t$. 

In both cases we clearly have $F_{t-1}BV_{t-1} = F_{t-1}N$, independently of management actions or market movements in $F_t$. 

Decision (1): 
using \eqref{e:MVa},  the second term in \eqref{e:ROA} is now calculated as
\revise{
\begin{align*}
F_{t-1}UG_{t-1} - \Delta E[UG_t|\mathcal{F}_{t-1}] 
&= F_{t-1}(MV_{t-1} - N) - E[MV_{t}|\mathcal{F}_{t-1}] + MV_{t-1} 
+ E[\textup{min}(MV_t,N)|\mathcal{F}_{t-1}] - N  \\
&= (K-F_{t-1})N
 - 
 E[\textup{max}((F_t - K)/(1+F_t),0) | \mathcal{F}_{t-1}] N .
\end{align*}
}
The third term gives 
\revise{
\begin{align*}
 ROA_t - E[ROA_t|\mathcal{F}_{t-1}] 
 &= cf_t^{a_1} + cf_t^{a_2} + \textup{min}((1+F_t)^{-1}(1+K)N, N) - N
 - KN   \\ 
 &\phantom{===} 
 - E[\textup{min}((1+F_t)^{-1}(1+K)N, N)|\mathcal{F}_{t-1}] + N  \\ 
 &=   \textup{min}((1+F_t)^{-1}(1+K)N, N)  
  - E[\textup{min}((1+F_t)^{-1}(1+K)N, N)|\mathcal{F}_{t-1}]
\end{align*}
}
since $cf_t^{a_1} = KN/2$ with respect to strategy (1) and $cf_t^{a_2} = KN/2$ in all cases. 

Decision (2): the second term remains unchanged since it corresponds to an $\mathcal{F}_{t-1}$ prediction. This term therefore captures the contribution of unrealized gains to book value return, $ROA_t$, that is due to the expected evolution of the portfolio. 
On the other hand, the third term becomes now, due to $cf_t^{a_1} = KN/2 + MV_t^{a_1}$ (coupon payments reward for having held the asset over the period $[t-1,t]$ and thus precede market placements) and $BV_t^{a_1} = 0$ (termination of asset), to 
\revise{
\begin{align*}
 ROA_t - E[ROA_t|\mathcal{F}_{t-1}] 
 &= 
 MV_t^{a_1} + \textup{min}((1+F_t)^{-1}(1+K)N, N)/2  
  - E[\textup{min}((1+F_t)^{-1}(1+K)N, N)|\mathcal{F}_{t-1}] \\
 &=  (1+F_t)^{-1}(1+K)N / 2 + \textup{min}((1+F_t)^{-1}(1+K)N, N)/2 \\
 &\phantom{===}
  - E[\textup{min}((1+F_t)^{-1}(1+K)N, N)|\mathcal{F}_{t-1}] .
\end{align*}
}
Notice that $MV_t^{a_1} =
(BV_t^{a_1})' + (UG_t^{a_1})'$ with
$(BV_t^{a_1})' 
= \textup{min}((1+F_t)^{-1}(1+K)N, N)/2$ where $(\cdot)'$  denotes \revise{the} asset's value before selling.
Hence this corresponds to a realization of unrealized gains, 
$(UG_t^{a_1})' = \textup{max}( (K-F_t)/(1+F_t) , 0 )N/2$, 
as a cash flow.
If $K\le F_t$ it follows that $(UG_t^{a_1})' = 0$,  and in this case the decision to sell does not have any effect on the return because of the lower of market value or cost principle.
The third term thus represents the unpredicted return (due to market movements or management actions at $t$).

\revise{
Thus this remark illustrates how market fluctuations  and  management decisions at $t$ influence the third term, $ROA_t - E[ROA_t|\mathcal{F}_{t-1}]$, in \eqref{e:ROA} while leaving the other two terms unaffected.  
The} above formulae can be simplified by  considering the special cases $F_t = 0$ \revise{or} $K=0$ with $F_t>0$. 
\end{remark}

\subsection{Contracts and model points}
We now describe the contracts which exist in the company's liability portfolio at valuation time. All contracts are assumed to have a minimum guaranteed benefit and bonus benefit which depends on the company's profit and is shared between policyholder, shareholder and tax office. Each contract gives rise to either a single maturity or mortality benefit, or multiple annuity benefits, or a surrender benefit, cost payments associated with the contract and premium payments \revise{accepted} by the company. We refer to \cite{Gerber} for a detailed exposition of life insurance mathematics.

Each contract, $c$, has at each time step, $t$, a mathematical reserve, $\widetilde{V}_t^c$, defined by actuarial principles and first order assumptions. Further, at valuation time the contract may have already received bonus declarations. The sum of these declared bonuses up to and including time $t=0$ are denoted by $(\widetilde{DB}_0^{\le0})^c$. If the contract is still active at $t>0$ this account remains unchanged, $(\widetilde{DB}_t^{\le0})^c=(\widetilde{DB}_0^{\le0})^c$. The sum of bonus declarations \revise{after valuation time} up to and including time $t>0$ is denoted by $\widetilde{DB}_t^c$, and we have $\widetilde{DB}_0^c = 0$.
The contracts total reserve is thus $\widetilde{LP}_t^c = \widetilde{V}_t^c + (\widetilde{DB}_t^{\le0})^c + \widetilde{DB}_t^c$.

If $m>0$ is the contract's maturity, the policyholder receives a guaranteed minimum benefit, $gbf_m^c$, plus declared bonuses, i.e., the benefit cash flow at $m$ equals 
$gbf_m^c + (\widetilde{DB}_{m-1}^{\le0})^c + \widetilde{DB}_{m-1}^c$. Profit that is generated at $m$ is not shared with contracts maturing at the same time. Annuity payments are analogous with a corresponding fraction of bonuses being paid out. If the policyholder dies at $0<t<m$ the death benefit similarly consists of a sum of a guaranteed minimum benefit, $gbf_t^c$, and declared bonuses, $(\widetilde{DB}_{t-1}^{\le0})^c + \widetilde{DB}_{t-1}^c$. 
Notice that, contrary to $\widetilde{DB}_{t-1}^c$, the cash flow due to $(\widetilde{DB}_{t-1}^{\le0})^c$ is already guaranteed at valuation time since these bonuses have already been declared. 
If the policyholder decides to surrender the contract at $0<t<m$ the resulting surrender benefit is $(1 - \chi_t^c)\widetilde{LP}_{t-1}^c$ where $0\le\chi_t^c\le1$ is the fraction which gives rise to the surrender gain made by the company.

For modelling purposes it is advantageous to describe model points instead of contracts. A model point is defined as either a single contract, or a collection of identical contracts, such that survival probabilities are already taken into account. Thus if $x$ is the model point associated to a contract $c$ and $p_0^t$ is the survival probability (incorporating mortality and surrender) between $0$ and $t$ the reserves  are related by $LP_t^x = p_0^t\widetilde{LP}_t^c$. However, $p_0^t$ may depend on the underlying economic scenario via dynamic policyholder behavior and therefore we refrain from introducing these probabilities explicitly. Rather, we let $V_t^x$\label{ref: V_t^x}, $(DB_t^{\le0})^x$\label{ref: (DB_t^le0)^x}, and $DB_t^x$\label{ref: DB_t^x} denote the  mathematical reserve, declared bonuses up to and including valuation time, and declared bonuses after valuation time, respectively, such that survival probabilities are taken into account. If policyholder behavior is dynamic with respect to economic scenarios then $V_t^x$ and $(DB_t^{\le0})^x$ are also scenario dependent. Future declared bonuses, $DB_t^x$, depend on economic scenarios by construction.  
Notice that $(DB_t^{\le0})^x\le (DB_0^{\le0})^x$, in general. Analogously, the probability weighted minimum guaranteed (maturity and mortality) benefits generated by model point $x$ at $t$ are denoted by $gbf_t^x$\label{ref: gbf_t^x}. The probability weighted cash flows, at $t$, due to $(DB_{t-1}^{\le0})^x$ are called $(gbf_t^{\le0})^x$\label{ref: (gbf_t^le0)^x}. Those due to $DB_{t-1}^x$ are called $ph_t^x$\label{ref: ph_t^x}. 

The total benefit cash flows are 
\[
 gbf_t = \sum_{x\in\mathcal{X}_t}gbf_t^x,
 \qquad 
 gbf_t^{\le0} = \sum_{x\in\mathcal{X}_t}(gbf_t^{\le0})^x,
 \qquad 
 ph_t = \sum_{x\in\mathcal{X}_t}ph_t^x 
\] \label{ref: gbf_t} \label{ref: gbf_t^le0} \label{ref: ph_t}
and the total reserves are given correspondingly by 
\[
 V_t = \sum_{x\in\mathcal{X}_t}V_t^x,
 \qquad 
 DB_t^{\le0} = \sum_{x\in\mathcal{X}_t}(DB_t^{\le0})^x,
 \qquad 
 DB_t = \sum_{x\in\mathcal{X}_t}DB_t^x 
\] \label{ref: V_t} \label{ref: DB_t^le0} \label{ref: DB_t}
where $\mathcal{X}_t$\label{ref: X_t} denotes the set of model points active at time $t$.

\subsection{Gross surplus and profit sharing}\label{sec: gs}
Let $L_t$ \label{ref: L_t} denote the book value of liabilities at time $t$, and we assume that $L_t = 
LP_t + SF_t$ is a sum of two items:
firstly, the life assurance provision, $LP_t = V_t + DB_t^{\le0} + DB_t$\label{ref: LP_t}; and secondly,  the surplus fund, $SF_t$\label{ref: SF_t}.
The surplus fund at time $t$, $SF_{t}$, consists of those profits that have not yet been declared to policyholders. As opposed to $LP_t$, $SF_t$ belongs to the collective of policyholders and cannot be attributed to individual contracts. 
This set-up follows the same logic as \cite{Gerstner08} where $V_t$, $DB_t^{\le0}+DB_t$, and $SF_t$ are referred to as the actuarial reserve,  allocated bonus, and free reserve (buffer account), respectively.  

The difference between total book value of assets and liabilities \revise{is} the free capital, 
$
    BV_t = FC_t + L_t 
$.\label{ref: FC_t}
Since \revise{return on free capital is not shared with  policyholders and therefore does} not contribute to the future discretionary benefits that we are interested in, we assume without loss of generality that $FC_t = 0$, so that the initial book value of assets is equal to the initial value of liabilities. Further, we assume that all shareholder gains that are produced by the company over the projection time are directly paid out to shareholder, and not accumulated within the company so that $FC_t=0$  (cf.~\cite[A.~2.2]{HG19}):
\begin{equation} \label{e:BV = LP + SF}
    BV_t 
    = L_t 
    = LP_t+SF_t
\end{equation}
Indeed, for purposes of best estimate calculation the free capital is not relevant as this, and the corresponding revenue, is not shared with policyholders. Moreover, assets used to cover statutory reserves may not be attributed separately to $L_t$ and $FC_t$ (\cite{MindZV,GBVVU}). Hence setting $FC_t=0$  leads to the appropriate scaling of revenue that is to be shared with policyholders. Thus the equality \eqref{e:BV = LP + SF} can always be achieved by replacing $BV_0$, $MV_0$, $UG_0$ by $BV_0' = BV_0-FC_0$, $MV_0' = \frac{BV_0'}{BV_0}MV_0$, $UG_0' = \frac{BV_0'}{BV_0}UG_0$, respectively, and assuming that future shareholder gains are paid out as cash flows which leave the model. 
On the other hand, if the goal is to build a comprehensive asset liability model to simulate the shareholder's point of view then the free capital is of paramount importance. However, in the simulation of a run-off liability book as under Solvency~II such a point of view is difficult to realize since the relation between $FC_t$ and $L_t$ quickly becomes unrealistic (without introducing new business). 
Notice also that the equity position in the balance sheet model of \cite{Gerstner08} is a hybrid of free capital, in the above sense of $FC_t$, and hidden reserves, $UG_t$. We do retain $UG_t$ in the projection since this is indispensable for best estimate calculation. 

The profit sharing mechanism depends on the company's gross surplus with respect to local accounting rules. The gross surplus can be described verbally as the sum of book value return, increase or decrease of statutory reserves, and all relevant cash flows (premiums, benefits, costs).     
The gross surplus, $gs_t$\label{ref: gs_t}, at $t$ is therefore defined as  
\begin{align*}
    gs_t 
    :=
    ROA_t
    - \Delta V_t - \Delta DB_t^{\le0} - DB_t^- + DB_{t-1}
    + pr_t - gbf_t - gbf_t^{\le0} - ph_t - co_t
\end{align*}
where $ROA_t$ is the book value return \eqref{e:ROA}, $DB_t^-$\label{ref: DB_t^-} is the account of declared bonuses before bonus declaration at $t$, $pr_t$\label{ref: pr_t} are premium payments and $co_t$\label{ref: co_t} are all cost cash flows. If $\chi_t$\label{ref: chi_t} is the appropriately averaged surrender fee factor at $t$, the discrepancy between $- DB_t^- + DB_{t-1}$ and $ph_t$ can be expressed as $- DB_t^- + DB_{t-1} - ph_t = \chi_t DB_{t-1}$. The analogous expression holds with an appropriately chosen factor $\chi_t^{\le0}$, representing surrender fees:   \revise{$-\Delta DB_t^{\le0} = \chi_t^{\le0} DB_{t-1}^{\le0}$}. Let further $\rho_t$\label{ref: rho_t} denote the averaged technical interest rate at $t-1$ such that  $-\Delta V_t + pr_t - gbf_t - co_t = -\rho_t V_{t-1} \revise{-} tg_t$ where $tg_t$\label{ref: tg_t} are the technical gains due to mortality and cost margins and surrender fees stemming from the mathematical reserves. We define $\gamma_t := (tg_t + \chi_t^{\le0} DB_{t-1}^{\le0} + \chi_t DB_{t-1})/LP_{t-1}$\label{ref: gamma_t} and express the gross surplus as 
\begin{equation}
\label{e:gs_def}
     gs_t 
 = ROA_t - \rho_t V_{t-1} + \gamma_t LP_{t-1}.
\end{equation} 

\begin{remark}
The advantage of expressing the gross surplus in this form is that the effects of book value return on assets, guaranteed technical interest rate and technical gains (i.e., mortality, cost and surrender margin) can be isolated.
\end{remark}

If $gs_t$ is positive, it is the surplus shared between policyholder, shareholder and tax office. If it is negative it is covered by the shareholder. 
Indeed, profit sharing is defined by legislation (\cite{MindZV,GBVVU, Dorobantu_etal20}) and requires that $gs_t$ is  shared between shareholders, policyholders and tax office according to
\begin{equation}
\label{e:ph*}
    gs_t
    = sh_t + ph_t^* + tax_t
\end{equation} \label{ref: sh_t} \label{ref: ph_t^*} \label{ref: tax_t}
where $sh_t = gsh\cdot gs_t^+ - gs_t^-$, $ph_t^* = gph\cdot gs_t^+$ and $tax_t = gtax\cdot gs_t^+$, and where $gsh$, $gph$ and $gtax$ are positive numbers such that $gsh + gph + gtax = 1$.   Furthermore, $c^+$ and $c^-$ denote the positive and negative parts of a number $c$, respectively. 
 
Notice that $sh_t$ and $tax_t$ constitute cash flows since this is money that leaves the company (and the model), while $ph_t^*$ is an accounting flow since this corresponds to a quantity that is transferred within the company to a different account but is not paid out as a benefit at time $t$. 

A fundamental principle of traditional life insurance is that profit sharing is \emph{not} equal to profit declaration (\cite{MindZV,GBVVU}). 
This means that $ph_t^* = gph\cdot gs_t^+$ is not necessarily declared (or credited) to its full extent to specific policyholder accounts.  
Rather, the profit sharing mechanism in traditional life insurance is such that a part, $\nu_t\cdot ph_t^*$ with $0\le\nu_t\le1$, of $ph_t^*$ is declared to  the policyholder accounts. Management may choose the value $\nu_t$ at each accounting step $t$. However, declaration may also be augmented by additional contributions, $\eta_t\cdot SF_{t-1}$ with $0\le\nu_t\le1$, from the previously existing surplus fund $SF_{t-1}$. again, management may choose the value $\eta_t$ at each accounting step $t$. Typically, surplus fund contributions will take place when $ph_t^*$ is small compared to management goals.

The total bonus declaration to $DB_t$ at time $t$ is therefore of the form
\begin{equation}
\label{e:decl}
    \nu_t\cdot ph_t^* + \eta_t\cdot SF_{t-1}
\end{equation} \label{ref: nu_t}
with the factor $\nu_t$ and $\eta_t$ determined according to management rules. 

As above, let $ph_t$ denote the amount of discretionary benefits paid out at time $t$. This cash flow depends on declarations to the declared benefits account, $DB_k$, which have occurred at times $0<k<t$. Declarations at valuation time, $t=0$, belong by definition to $DB_t^{\le0}$, and the resulting cash flows are already guaranteed at $t=0$, whence these do not contribute to $ph_t$. Therefore, we have 
\begin{equation}
    ph_1=0. 
\end{equation}
For $k>0$ let $0\le\eta_k\le1$\label{ref: eta_k} denote the amount of declaration from $SF_{k-1}$ to $DB_k$. The numbers $\eta_k$ and $\nu_k$ are, in general, unknown at $t=0$ and depend on management rules. Let further $0\le \mu_k^t\le1$\label{ref: mu_k^t} denote the fraction of bonus declarations at $k$, $\eta_k\cdot SF_{k-1} + \nu_k\cdot ph_k^*$, that is either paid out as a future discretionary benefit (in case of contract maturity or mortality) or kept by the company as a surrender fee (in case of premature contract termination), $sg_t^*$\label{ref: sg_t^*}, at time $t$. This fraction is also unknown and may depend, in general, on dynamic policyholder behavior. 
The defining relation is thus
\begin{equation}
    \label{e:ph_t-def}
    ph_t + sg_t^*
    = \sum_{k=1}^{t-1} \mu_k^t\Big(\eta_k\cdot SF_{k-1} + \nu_k\cdot ph_k^*\Big)
\end{equation}
where $t\ge2$.
Since the sum of discretionary benefits and surrender fees cannot exceed the amount of previous declarations we must have 
\begin{equation}
    \label{e:mu-le1}
    \sum_{t=k+1}^T\mu_k^t \le 1 . 
\end{equation}
\revise{To} sum up the above discussion, passing from $t-1$ to $t$, $DB_{t-1}$ is:  
\begin{itemize}
    \item 
    increased by declarations $\eta_t\cdot SF_{t-1}$ from the surplus fund
    where $0\le\eta_t\le1$ is chosen by the management,
    \item
    increased by direct policyholder declarations $\nu_t\cdot ph_t^*$
    where $0\le\nu_t\le1$ is chosen by the management,
    \item
    decreased by cash flows, $ph_t$, to policyholders whose contracts terminate at $t$, and
    \item
    decreased by accounting flows $sg_t^* := \chi_t\cdot DB_{t-1}$ with $0\le\chi_t\le1$ due to surrender fees. The fraction $\chi_t\cdot DB_{t-1}$ is freed up, in the sense that it is not attributed to specific contracts anymore, and thus contributes to the annual gross surplus.  
\end{itemize}
Therefore, we have the iterative relation
\begin{equation}
\label{e:DB0}
    DB_t
    = DB_{t-1}
    + \eta_t\cdot SF_{t-1}
    + \nu_t\cdot ph_t^*
    - ph_t
    - sg_t^*
\end{equation}
with starting point $DB_0 = 0$.

Consider the surplus fund $SF_{t-1}$ at $t-1$. Going one time step further, it is increased by allocating 
$
 (1-\nu_t)\cdot ph_t^*
$
to the fund, which is the part of $ph^*_t$ not declared to the policyholders' accounts, and decreased by declaring 
$
 \eta_t\cdot SF_{t-1}
$
to policyholder accounts. We thus obtain 
\begin{equation}
    \label{e:SF}
    SF_t
    = SF_{t-1} 
        + (1-\nu_t)\cdot ph_t^* 
        - \eta_t\cdot SF_{t-1}
\end{equation}
with known starting point $SF_0$. Together with \eqref{e:DB0} this yields
\begin{equation}
\label{e:DBSFevol}
    \Delta\,(DB_t+SF_t)
    =
    DB_t + SF_t - DB_{t-1} - SF_{t-1} 
    = 
    + ph_t^*
    - ph_t
    - sg_t^*,
\end{equation}
Crucially, the model dependent fractions $\nu_t$ and $\eta_t$ do not appear in this equation. 
This evolution equation for $DB_t + SF_t$ is the starting point for the subsequent analysis. The sum $DB_t+SF_t$ is increased at each time step by the total shared profit, $ph_t^*$, as opposed to the declaration~\eqref{e:decl}, and therefore represents the statutory reserves of previously shared profit, although the totality of this sum is not a (single) balance sheet item.

For further reference we observe that 
\begin{align}\label{e:ph1}
    PH^*
    &:= E\Big[\sum B_t^{-1} ph_t^*\Big] 
    = gph\cdot E\Big[\sum B_t^{-1} gs_t^+\Big]\\
    \notag
    &= gph\cdot  E\Big[\sum B_t^{-1} gs_t\Big] 
     + gph\cdot  E\Big[\sum B_t^{-1} gs_t^-\Big]\\
    &= gph\cdot  (VIF + PH^* + TAX ) + gph\cdot COG.
    \notag
\end{align} \label{ref: PH^*}
Here we have used the splitting $gs_t = sh_t + ph_t^* + tax_t$, where $sh_t$ and $tax_t$ are the shareholder and tax cash flows as define\revise{d} above, to obtain the value of in-force business
\[
    VIF 
    = E\Big[\sum_{t=1}^T B_t^{-1} sh_t\Big]
    = E\Big[\sum_{t=1}^T B_t^{-1} (gsh\cdot gs_t^+ - gs_t^-) \Big],
\] \label{ref: VIF}
the cost of guarantees
\begin{equation} 
\label{e:COG}
    COG 
    = E\Big[\sum_{t=1}^T B_t^{-1} gs_t^-\Big],
\end{equation} \label{ref: COG}
and
\[
TAX 
    = E\Big[\sum_{t=1}^T B_t^{-1} tax_t \Big]
    =  E\Big[\sum_{t=1}^T B_t^{-1} gtax\cdot gs_t^+ \Big].
\] \label{ref: TAX}

\subsection{Future discretionary benefits}
According to Solvency~II \cite{L1,L2}, the best estimate is the expectation of all future cash flows which are related to existing business. 
These cash flows are benefits, $gbf_t+gbf_t^{\le0}+ph_t$, premium income, $pr_t$, and costs, $co_t$. The best estimate is thus defined as
\begin{equation}
BE 
 := 
 E\left[\sum_{t=1}^TB_t^{-1}(gbf_t+gbf_t^{\le0}+ph_t + co_t - pr_t)\right]
\end{equation} \label{ref: BE}
The value of future discretionary benefits is given by
\begin{equation}
\label{e:fdb-def}
FDB := E\left[\sum_{t=1}^T B_t^{-1}ph_t\right]
\end{equation} \label{ref: FDB}
and the value of the guaranteed benefits is by definition  
\begin{equation}
\label{ref: GB}
    GB := BE - FDB 
    = 
    E\Big[\sum_{t=1}^TB_t^{-1}(gbf_t + gbf_t^{\le0} + co_t - pr_t)\Big] 
\end{equation} 
If actuarial variables are independent from economical ones, then $gbf_t + gbf_t^{\le0}$, $co_t$ and $pr_t$, and hence $GB$, may be calculated from a purely deterministic model.
This independence would exclude the possibility of dynamic policyholder behavior (e.\,g., surrender depends dynamically on a comparison of declared bonuses and the prevailing yield curve).
However, for our results to hold, we do not need to make this assumption. 

\begin{lemma}
\begin{equation}\label{e:PH*}
    PH^*
    = 
    \frac{gph}{1-gph}\Big(MV_0 - E[B_T^{-1}MV_T] - GB - FDB + COG \Big)
\end{equation}
\end{lemma}

\begin{proof}
We use the no-leakage principle \cite[Prop.~2.2]{HG19} which states that $MV_0 = BE + VIF + TAX + E[B_T^{-1}MV_T]$. With the decomposition \revise{$BE=GB+FDB$} and definition~\eqref{e:ph1} this implies that 
$PH^* = gph(PH^* + MV_0 - GB - FDB - E[B_T^{-1}MV_T] + COG)$. 
\end{proof}

\begin{remark}\label{rem:no-leak}
The no-leakage principle \cite[Prop.~2.2]{HG19} essentially states that in a risk neutral model all cash flows have to be accounted for and all future expected gains or losses have to be reflected in the initial market value. This is a general statement and uses only no arbitrage theory and the generally accepted accounting principles which define the cash flows leading to the quantities $BE$, $VIF$ and $TAX$. However, the precise formulation of the accounting principles is not relevant in this context since it suffices that cash flows are well defined and that there can be no other cash flows except those to the policyholder (including costs), to the shareholder and to the tax office.  
\end{remark}

\section{A representation of $FDB$}\label{sec:fdb_rep}
Equation~\eqref{e:DBSFevol} may be rephrased as
$
    B_t^{-1}\Delta (DB_t + SF_t) 
    = B_t^{-1}ph_t^*  - B_t^{-1}ph_t - B_t^{-1}sg_t^*
$.
By virtue of $\Delta(f_t g_t) = (\Delta f_t) g_{t-1} + f_t\Delta g_t$,
we obtain a discrete integration by parts formula 
\begin{align}
\label{e:ibp}
   \sum_{t=1}^T \Big( B_t^{-1}ph_t^*  - B_t^{-1}ph_t - B_t^{-1} sg_t^* \Big) 
   &= 
   \sum_{t=1}^T B_t^{-1}\Delta (DB_t + SF_t)  \\ 
   \notag 
   &= 
   \sum_{t=1}^T \Delta ( B_t^{-1} (DB_t + SF_t) )
   - \sum_{t=1}^T  (DB_{t-1} + SF_{t-1}) \Delta B_t^{-1} \\ 
   \notag 
   &= 
   B_T^{-1}(DB_T + SF_T) - SF_0
   + \sum_{t=1}^T  (DB_{t-1} + SF_{t-1}) F_{t-1} B_t^{-1}
\end{align}
since $\Delta B_t^{-1} = -F_{t-1}B_t^{-1}$. 
Taking the expectation of this equation and using \eqref{e:PH*}, we find 
\begin{align*}
    &\frac{gph}{1-gph}\Big(MV_0 - E[B_T^{-1}MV_T]-GB+COG\Big) 
    - \frac{1}{1-gph}FDB \\
    &= PH^* - FDB \\
    &= 
    E\Big[ B_T^{-1}(DB_T + SF_T) \Big] - SF_0
    + E\Big[ \sum_{t=1}^T  (DB_{t-1} + SF_{t-1}) F_{t-1} B_t^{-1} \Big] 
    + E\Big[ \sum_{t=1}^T B_t^{-1} sg_t^* \Big] 
\end{align*}
Because of $MV_0 = BV_0 + UG_0 = LP_0 + SF_0 + UG_0$, rearranging yields:
\begin{theorem}
\begin{equation}\label{e:fdb-rep}
    FDB 
    = 
    SF_0 + gph\Big(LP_0 + UG_0 - GB \Big) 
    + gph\cdot COG - I - II - III 
\end{equation}
where
\begin{align*}
    I
    &:= E\Big[B_T^{-1}\Big(DB_T+SF_T + gph(UG_T + V_T + DB_T^{\le0}) \Big)\Big] \\
    II
    &:= 
    (1-gph)E\Big[ \sum_{t=2}^T B_t^{-1} sg_t^*\Big] \\
    III
    &:=
    (1-gph) E\Big[ \sum_{t=1}^T F_{t-1} B_t^{-1}  (DB_{t-1} + SF_{t-1}) \Big] 
\end{align*}
\end{theorem}

\begin{remark}
The estimation formula~\eqref{e:fdb-rep} is derived without any model specific assumptions and relies therefore only on general accounting rules and the application of the no-leakage principle in \eqref{e:PH*}.
The `integration by parts' \eqref{e:ibp} transfers the problem of calculating $E[\sum B_t^{-1}\Delta(DB_t+SF_t)]$ to one of evaluating or \emph{estimating} the `boundary term' $E[B_T^{-1}(DB_T+SF_T)]-SF_0$ as-well as $II$ and $III$. The idea is now that these approximations should be feasible since $SF_0$ is known, $E[B_T^{-1}(DB_T+SF_T)]$ is expected to be negligible at run-off time $T$, and estimation errors in $II$ and $III$ concern only the the  surrender gains from future declared bonuses and the return, due to $F_{t-1}$, on $DB_{t-1} + SF_{t-1}$, respectively.  
\end{remark}


\begin{remark}
The interpretation of the constituents of \eqref{e:fdb-rep} is as follows: 
\begin{itemize}
\item 
Term $SF_0$ is not multiplied by $gph$. This makes sense since the surplus fund, while not assigned to individual contracts, already belongs to the policyholder collective (compare \cite[Article91]{L1}). There cannot be a transfer of funds from $SF_t$ to the shareholder or tax office. On the other hand, the return on $SF_t$ is shared between all parties, whence the corresponding deduction in term~$III$. 
\item Term $gph\cdot (LP_0 - GB)$: 
According to the local GAA Principle of prudentiality (e.g., \cite[§148(1)]{VAG}), the life assurance provisions are determined with respect to safety margins and  $gph\cdot(LP_0 - GB)$ represents the policyholder share of this margin.
\item 
Term $gph\cdot UG_0$ represents the policyholder share in the unrealized gains existing in the portfolio at valuation time. 
\item 
When the safety margins considered in the calculation of $LP_0$ are not sufficient (e.g., due to a very low interest rate environment) then cost of guarantees arise and manifest as shareholder capital injections. When the environment is such that losses are expected for all future valuation dates $1\le t\le T$ then the company's management could choose to inject just enough shareholder capital to cover these losses so that $COG$ balances the right hand side of \eqref{e:fdb-rep} to yield $FDB=SF_0-I-II-III$. 
Balancing the right hand side of \eqref{e:fdb-rep} would mean to realize hidden reserves, $UG_t$, before injecting new capital, and while this might be a realistic assumption, equation~\eqref{e:fdb-rep}  holds independently of all such management rules. 
Moreover, in practice it is difficult to determine this minimal amount precisely such that the possibly counter-intuitive appearance of $COG$ in \eqref{e:fdb-rep} represents the policyholder share of excess capital injections in $gph(LP_0 + UG_0 - GB + COG)$.
%
\item Term~$I$ is related to the policyholder share of assets that remain
in the company after run-off of the liability portfolio;
\item
Term~$II$ is the tax and shareholder share (since $1-gph = gsh + gtax$) in the gross surplus due to the fraction of declared future profits, $DB_t$, that is freed up because of  surrender fees. 
\item 
Term~$III$ captures the tax and shareholder shares in interests on allocated profits as well as on the surplus fund.
\end{itemize}
\end{remark}

\section{Assumptions}\label{sec:assump}

\subsection{Liability run-off assumptions} 

\begin{ass}
    \label{ass:runoff}
    The projection horizon $T$\label{ref: T} corresponds to the run-off time of the liability portfolio such that 
    $SF_T = LP_T = UG_T = 0$
    (cf.~\cite[A.~3.13]{HG19}).
\end{ass}

\begin{ass}
    \label{ass:geom}
    The expected life assurance provisions $E[LP_t]$ decrease geometrically: there is a fixed $1\le h < T$\label{ref: h} such that $E[LP_t] = l_t^h\, LP_0$ where $l_t^h := 2^{-t/h}$ for $t<T$ and $l_T^h := 0$. 
\end{ass}

Since the portfolio is in run-off there is a time, $h$, where $E[LP_h] = LP_0/2$. Continuing from $h$ onwards there has to be a time, $h+h'$, such that $E[LP_{h+h'}] = E[LP_h]/2$. Assuming that the company's business model has been stable over time we have time  homogeneity in the sense that $h'=h$ and run-off of the liability book is geometric. This would not be satisfied if the company under consideration has taken up business only very recently but for companies with a longer history we view this as a very good approximation. 

\begin{ass}
\label{ass:sigma}
In expectation the total declared bonuses are a fixed fraction of the life assurance provisions: $E[DB_t^{\le0} + DB_t] = \sigma E[LP_t]$ for all $0\le t\le T$ and a fixed $0\le\sigma \le 1$. 
Moreover, $E[DB_t^{\le0}]$ does not vanish too quickly: $E[DB_t]\le \sigma_t E[LP_t]$ where $\sigma_t :=  t\sigma/h$ for $t\le h$ and $\sigma_t := \sigma$ for $t>h$.
\end{ass}

\begin{ass} \label{ass:boundSF} 
The relation $SF_0/LP_0 =: \vartheta$\label{ref: theta} remains constant in expectation: $E[SF_t] = \vartheta E[LP_t]$ for all $0\le t\le T$.  (Cf.~\cite[A.~3.10]{HG19})
\end{ass}

Assumptions~\ref{ass:sigma} and \ref{ass:boundSF} are also statements about time-homogeneity. Management rules concerning bonus declarations should remain reasonably constant in the long run such that $\sigma$ and $\vartheta$ do not vary too strongly. The relevant point in this context is that these quantities should not vary arbitrarily but follow from target rates set by management rules. Assumption~\ref{ass:boundSF} is comparable to the assumption concerning the `annual interest rate' in \cite[Section~4.2]{Gerstner08}.

\subsection{Surrender assumption}

\begin{ass}
    \label{ass:surr}
    The surrender gains, $sg_t^* = \chi_t DB_{t-1}$, can be estimated on average with the same factor, $\gamma_t$, as the technical gains in \eqref{e:gs_def}: 
    $E[sg_t^*] \le E[\gamma_t DB_{t-1}]$.  
\end{ass}
The factor $\gamma_t$ comprises mortality, cost and surrender margins as a fraction of the full life assurance provision, $LP_{t-1}$. It is therefore reasonable to expect that the same factor can be used as an upper bound on the surrender margin arising from declared bonuses, $DB_{t-1}$, alone.

\subsection{Bonus benefit assumptions}
Because of equation~\eqref{e:SF} the bonus benefit declaration at $t$ can be expressed as $\eta_t\cdot SF_{t-1} + \nu_t\cdot ph_t^*  = SF_{t-1} - SF_t + ph_t^*$. Management rules generally strive to keep profit declarations stable while, in accordance with assumptions~\ref{ass:boundSF} and \ref{ass:geom}, $SF_t$ is expected to decrease geometrically over time. In order for $SF_t$ to decrease in expectation, the bonus benefit declarations must be strictly positive in expectation. To achieve this, a fraction of the profit share, $ph_t^*$, must also be declared to policyholders, at least in expectation. We turn this reasoning into an assumption along all scenarios. 

\begin{ass}
\label{ass:nu}
There is a fixed $0<\nu<1$\label{ref: nu} such that the declarations satisfy
$
\eta_t\cdot SF_{t-1} + \nu_t\cdot ph_t^* 
\ge \nu\cdot ph_t^*
$
for all $1\le t\le T$.
\end{ass}

\begin{ass}
\label{ass:mu}
Assume that $\mu_k^{s+1}$ is determined by the geometric run-off assumption~\ref{ass:geom}: 
$
    \mu_k^{s+1} 
    = \frac{l_s^h - l_{s+1}^h}{l_k^h}
$.
\end{ass}
Notice that, for fixed $k$, this definition entails $\sum_{s=k}^{T-1} \mu_k^{s+1} = 1-l_T^h/l_k^h = 1$; cf.~\eqref{e:mu-le1}. That is, run-off is complete at $T$.

\subsection{Gross surplus assumptions}
According to \eqref{e:gs_def} and \eqref{e:ROA}, the gross surplus is given by 
\begin{align*}
    gs_t
    =
    F_{t-1} (LP_{t-1} + SF_{t-1}) 
    + F_{t-1}UG_{t-1} - \Delta UG_t|\mathcal{F}_{t-1} 
    + ROA_t - ROA_t|\mathcal{F}_{t-1}
    - (\rho_t-\gamma_t) V_{t-1}
\end{align*}
where $\rho_t$ and $\gamma_t$ may, in general, also depend on the stochastic interest rate curve via dynamic surrender. 
 
For the purpose of estimating terms $III$ and $COG$ in \eqref{e:fdb-rep}, we make the following simplifying assumptions. The principle idea behind these assumptions is that the main source of stochasticity in $gs_t$ is the forward rate $F_{t-1}$ whence all other quantities are replaced by their expected values.  The simplified model of $gs_t$ will be denoted by $\widehat{gs}_t$.

\begin{ass}
    \label{ass:tech}
    In $\widehat{gs}_t$ the technical interest rate $\rho_t$ and the technical gains $\gamma_t$ are deterministic functions of $t$. 
\end{ass}

\begin{ass}
\label{ass:ROA}
In $\widehat{gs}_t$ the return $ROA_t$ is predictable, i.e.\ $\mathcal{F}_{t-1}$-measurable, and realizations of unrealized gains are determined by a fixed number $1< d < T$\label{ref: d}:
\begin{enumerate}[\up (1)]
    \item 
    $ROA_t - \revise{E[ROA_t|\mathcal{F}_{t-1}]} = 0$;
    \item 
    $F_{t-1} UG_{t-1} - \Delta UG_t|\mathcal{F}_{t-1} = P(0,t)^{-1}(l_{t-1}^d - l_t^d)UG_0$
    where $l_t^d := 2^{-t/d}$ for $t<T$ and $l_T^d:=0$;
\end{enumerate}
\end{ass}

The motivation for item (2) is as follows. The quantity $F_{t-1} UG_{t-1} - \revise{E[\Delta UG_t|\mathcal{F}_{t-1}]} =: cf_t^{UG}$ may be viewed as a cash flow due to realizations of unrealized gains as assets approach their maturities (an example of this reasoning is contained in Remark~\ref{rem:roaEx}). Indeed, if an asset $a$ has maturity $T_a$, then we must have $UG_{T_a}^a = MV_{T_a}^a - BV_{T_a}^a = 0$, and thus $UG_t^a$ tends to $0$ as $t$ approaches $T_a$. As a proxy for the number $d$ we take the duration of the portfolio. Assuming that cash flows, $cf_t^{UG}$, due to realizations of unrealized gains are known at valuation time $t=0$, we obtain $UG_0 = \sum_{t=1}^T P(0,t)cf_t^{UG}$. Setting $\sum_{t=1}^T P(0,t)(F_{t-1} UG_{t-1} - \revise{E[ \Delta UG_t|\mathcal{F}_{t-1})]} =
\sum_{t=1}^T P(0,t)cf_t^{UG} = UG_0 = \sum_{t=1}^T(l_{t-1}^d - l_t^d)UG_0$ and insisting on equality of the summands leads to the above assumption.

\begin{ass}\label{ass:BVvar}
The coefficient of variation of book valued items is negligible in comparison to that of market movements. Concretely, the coefficients of variations of $DB_t$, $LP_t$ and $SF_t$ are assumed to be negligible in comparison to to those of $F_t$ and $B_t^{-1}$. 
\end{ass}

This assumption reflects the general principle that book values are expected to be more stable than market values since not all market movements are reflected in book values but rather lead to unrealized gains or losses (\cite{Dorobantu_etal20}). 

Invoking the above assumptions \ref{ass:geom}, \ref{ass:ROA}, \ref{ass:BVvar}, \ref{ass:tech}, \ref{ass:sigma} and \ref{ass:boundSF}, we define
\begin{align}
\label{e:gshat}
    \widehat{gs}_t 
    :&= 
    F_{t-1} E[BV_{t-1}] 
    + P(0,t)^{-1}(l_{t-1}^d - l_t^d)UG_0
    - \rho_t V_{t-1} + \gamma_t LP_{t-1} \\
    \notag 
    &=
    \Big(
    F_{t-1} 
    + P(0,t)^{-1}\frac{l_{t-1}^d-l_t^d}{l_{t-1}^h}\frac{UG_0}{(1+\vartheta)LP_0}
    -\frac{(1-\sigma)\rho_t-\gamma_t}{1+\vartheta}
    \Big) (1+\vartheta)l_{t-1}^h LP_0 
\end{align}
to be used as a simplified model for $gs_t$.

\section{Analytical lower and upper bounds for future discretionary benefits}\label{sec:LBandUB}
The model dependent quantities in \eqref{e:fdb-rep} are $I$, $II$, $III$ and $COG$. 
Calculating these explicitly is just as difficult as calculating the $FDB$. The purpose of this section is therefore to derive analytical bounds for these quantities, i.e.\ bounds which can be calculated without a numerical model. 

\subsection{Estimating $I$}
In accordance with Assumption \ref{ass:runoff}, we estimate $I$ by
\begin{equation}
\label{e:Ihat}
    \widehat{I} = 0.
\end{equation}
Compare also with the second statement in \cite[Prop.~2.2]{HG19}.

\subsection{Estimating $II$}
The expression $sg_t^* = \chi_{t} DB_{t}$ in Term~$II$ corresponds to the fraction of $DB_{t}$ that is freed up each year due to policyholder surrender fees and thus contributes to the company's surplus as a component of the surrender gains.
Assumptions~\ref{ass:BVvar}, \ref{ass:surr}, \ref{ass:sigma} and \ref{ass:geom} imply that term $II$ can be estimated as $II\le\widehat{II}$ with
\begin{align}
\label{e:II_hat}
    \widehat{II} 
    := (1-gph)\sum_{t=2}^T\gamma_t\sigma_t P(0,t)l_{t-1}^h LP_0.
\end{align}

\begin{remark}
\label{rem:surrII}
The product $(1-gph)\gamma_t\sigma_t$ is expected to be very small since this represents the shareholder and tax share of the surrender gains from future declared bonuses, such that $\widehat{II}$ should be also very small in comparison to $MV_0$.  
\end{remark}

\subsection{Bounding $III$ from above}
The essential idea is to use the recursive relation \eqref{e:DBSFevol} to obtain an upper bound for $DB_t+SF_t$.  Equation~\eqref{e:DBSFevol} implies 
\begin{align}
\label{e:recur}
    DB_t + SF_t 
    &= SF_0 + \sum_{s=1}^t ph_s^* - \sum_{s=2}^t (ph_s+sg_s^*) 
    = 
    SF_0 + gph\cdot gs_t^+
    + \sum_{s=1}^{t-1}\Big( 
     gph\cdot gs_s^+ - ph_{s+1} - sg_{s+1}^*
    \Big).
\end{align}
Assumptions~\ref{ass:nu} and \ref{ass:mu} yield 
\begin{align*}
    \sum_{s=1}^{t-1}\Big(ph_{s+1}+sg_{s+1}^*\Big)
    &= 
    \sum_{s=1}^{t-1}\sum_{k=1}^s\mu_k^{s+1}\Big(\eta_k\cdot SF_{k-1} + \nu_k\cdot ph_k^*\Big)\\
    &\ge 
    \nu \sum_{s=1}^{t-1}\sum_{k=1}^s\mu_k^{s+1}ph_k^* 
    = 
    \nu\sum_{k=1}^{t-1}\sum_{s=k}^{t-1}\frac{l_s^h-l_{s+1}^h}{l_k^h}ph_k^* 
    = 
    \nu\sum_{k=1}^{t-1}\frac{l_k^h-l_{t}^h}{l_k^h}ph_k^*
\end{align*}
whence \eqref{e:recur} satisfies 
\begin{equation}
    DB_t+SF_t 
    \le 
    SF_0 + gph\cdot gs_t^+ 
    + gph\cdot \sum_{s=1}^{t-1}
      \Big( 
       1 - \nu (1 - l_{t-s}^h )
      \Big) gs_s^+ .
\end{equation}
Thus 
\begin{align}
\label{e:III_1}
    III
    &=
    (1-gph)\sum_{t=0}^{T-1} E\Big[ B_{t+1}^{-1}F_{t}(DB_{t}+SF_{t}) \Big]  \\ 
    \notag
    &\le
    (1-gph) (1-P(0,T) ) SF_0 
    + (1-gph)gph \sum_{t=1}^{T-1} E\Big[B_{t+1}^{-1}F_{t} \cdot gs_t^+  \Big] \\
    \notag 
    &\phantom{===}
    + (1-gph)gph \sum_{t=2}^{T-1} \sum_{s=1}^{t-1} 
    \Big( 
       1 - \nu (1 - l_{t-s}^h )
      \Big)
    E\Big[B_{t+1}^{-1}F_{t}  
       \cdot gs_s^+ \Big] 
\end{align}
The expression $E[B_{t+1}^{-1}F_{t} \cdot gs_s^+ ]$ gives the fair value of the risk free return on $gs_s^+$ in the period from $t$ to $t+1$.

\subsection{Bounding $III$ from below} 
We use again equation~\eqref{e:recur}, and notice that equations~\eqref{e:ph_t-def} and \eqref{e:SF} imply
\begin{align*}
    \sum_{s=2}^t\Big( ph_s + sg_s^*\Big) 
    &= 
    \sum_{s=2}^t 
     \sum_{k=1}^{s-1} \mu_k^s\Big(SF_{k-1} - SF_k + ph_k^*\Big) %
    =
    \sum_{k=1}^{t-1}
     \sum_{s=k+1}^t \mu_k^s \Big(  
     SF_{k-1} - SF_k + ph_k^*
    \Big)
    \\
    &\le 
     SF_0 - SF_{t-1} + \sum_{k=1}^{t-1}ph_k^*  
\end{align*} 
since $\sum_{s=k+1}^t \mu_k^s\le 1$ due to \eqref{e:mu-le1}. 
Inserting this in \eqref{e:recur} yields  
$DB_t + SF_t \ge ph_t^* + SF_{t-1}$ for all $t\ge1$,
and therefore 
\begin{align*}
  III 
  &= (1-gph)E\Big[\sum_{t=0}^{T-1} F_t B_{t+1}^{-1}(DB_t+SF_t)\Big] \\
  \notag
  &\ge 
  (1-gph) F_0(1+F_0)^{-1} SF_0 +
  (1-gph)E\Big[\sum_{t=1}^{T-1} F_t B_{t+1}^{-1}
   \Big( gph\cdot gs_t^+ + SF_{t-1} \Big)\Big] .
\end{align*}
This estimate does not depend on any of the assumptions in Section~\ref{sec:assump}. Neglecting, in accordance with Assumption~\ref{ass:BVvar}, the variation of $SF_{t-1}$ in comparison to that of $F_{t-1}$, and using Assumptions~\ref{ass:boundSF} and \ref{ass:geom}, yields
\begin{align}
  \label{e:III_lb1}
  III 
  &\ge 
  (1-gph)\Big( 
   F_0(1+F_0)^{-1} SF_0 +
    \vartheta\sum_{t=1}^{T-1}(P(0,t)-P(0,t+1) )l_{t-1}^h LP_0
   \Big) \\ 
   \notag
   &\phantom{===}
   +
  gph(1-gph)E\Big[\sum_{t=1}^{T-1} F_t B_{t+1}^{-1} gs_t^+ \Big] .
\end{align}

\subsection{Estimating the return on the deferred caplet}
Let us rewrite 
$F_{t}B_{t+1}^{-1}  gs_s^+  
= (B_{t}^{-1}-B_{t+1}^{-1})   gs_s^+ 
= (D(s,t)-D(s,t+1))  B_s^{-1} gs_s^+$ and 
abbreviate the coefficients of variations as
\begin{equation}
    CV_{s,t}^1 :=
    CV\Big[D(s,t)-D(s,t+1)\Big] ,\qquad
    CV_s^2 :=
    CV\Big[ B_s^{-1} gs_s^+\Big]
    .
\end{equation} \label{ref: CV_s,t^1} \label{ref: CV_s^2}
Since $-1 \le Corr_s^t  := Corr[D(s,t)-D(s,t+1), B_s^{-1} gs_s^+] \le 1$, it follows that 
\begin{align*}
    E\Big[(D(s,t&)-D(s,t+1))  B_s^{-1} gs_s^+ \Big]
    = 
    E\Big[(D(s,t)-D(s,t+1))\Big] E\Big[ B_s^{-1} gs_s^+ \Big] \\
    &\phantom{===}
    + Corr_s^t
        \cdot CV_{s,t}^1 
              CV_s^2
        E\Big[(D(s,t)-D(s,t+1))\Big] E\Big[ B_s^{-1} gs_s^+ \Big]\\ 
    &\le 
    \Big(P(s,t) - P(s,t+1)\Big) E\Big[ B_s^{-1} gs_s^+ \Big]
    \Big( 1 + CV_{s,t}^1 CV_s^2  \Big)
\end{align*}
and 
\begin{equation*}
     E\Big[(D(s,t)-D(s,t+1))  B_s^{-1} gs_s^+ \Big]
     \ge 
     \Big(P(s,t) - P(s,t+1)\Big) E\Big[ B_s^{-1} gs_s^+ \Big]
    \Big( 1 - CV_{s,t}^1 CV_s^2  \Big)
\end{equation*}
Hence \eqref{e:III_lb1} and \eqref{e:III_1} lead to 
\begin{align}
\notag
  &(1-gph)\Big( 
   F_0(1+F_0)^{-1} SF_0 +
    \vartheta\sum_{t=1}^{T-1}(P(0,t)-P(0,t+1) )l_{t-1}^h LP_0
   \Big) \\ 
   \notag
   &\phantom{===}
   +
  gph(1-gph)\sum_{t=1}^{T-1}
  \Big(1 - CV_{0,t}^1\, CV_t^2 \Big)
      \Big(1 - P(t,t+1)\Big)
  E\Big[ B_{t}^{-1} gs_t^+ \Big] 
  \\
  \label{e:III_2}
    &\le III \\
    \notag 
    &\le
    (1-gph) (1-P(0,T)) SF_0 \\
    \notag 
    &\phantom{===}
    + (1-gph)gph \sum_{t=1}^{T-1} 
    \Big(1 + CV_{0,t}^1\, CV_t^2 \Big)
    \Big( 1 - P(t, t+1) \Big)
    E\Big[B_{t}^{-1}  gs_t^+  \Big] \\
    \notag 
    &\phantom{===}
    + (1-gph)gph \sum_{t=2}^{T-1} \sum_{s=1}^{t-1} 
    \Big( 
       1 - \nu (1 - l_{t-s}^h )
      \Big)
      \Big(1 + CV_{s,t}^1\, CV_s^2 \Big)
      \Big(P(s,t) - P(s,t+1)\Big) E\Big[ B_s^{-1} gs_s^+ \Big]
\end{align}
The term $E[ B_s^{-1} gs_s^+ ]$ is the value of the caplet with payoff $gs_s^+$ at time $s$. 

\subsection{Estimating the caplet}
Now we replace $gs_s$ by its simplified model $\widehat{gs}_s$ defined in \eqref{e:gshat}. This implies that the (simplified) caplet can be expressed as 
\begin{equation}
    E\Big[ B_s^{-1} \widehat{gs}_s^+ \Big]
    = 
    \mathcal{O}_s^+ (1+\vartheta)l_{s-1}^h LP_0 
\end{equation}
where 
\begin{equation}
    \label{e:O}
 \mathcal{O}_s^{\pm}
 := E\Big[B_s^{-1}\Big(
   F_{s-1} 
    + P(0,s)^{-1}\frac{l_{s-1}^d-l_s^d}{l_{s-1}^h}\frac{UG_0}{(1+\vartheta)LP_0}
    - \frac{(1-\sigma)\rho_t-\gamma_t}{1+\vartheta}
 \Big)^{\pm}\Big]
\end{equation} \label{ref: O_s}
is the value at $0$ of the caplet (corresponding to $+$) or floorlet (corresponding to $-$) with maturity $s-1$ and payment
payment  
$(
   F_{s-1} 
    + P(0,s)\frac{l_{s-1}^d-l_s^d}{l_{s-1}^h}\frac{UG_0}{(1+\vartheta)LP_0}
    - \frac{(1-\sigma)\rho_t-\gamma_t}{1+\vartheta}
 )^{\pm}$ occurring at settlement date $s$. In the normal model this value is given by the Black formula
 (\cite{Black76,BM06})
\begin{equation} 
\label{e:Black} 
    \mathcal{O}_s^{\pm} 
    = 
    P(0,s)\cdot\Big(
     \pm ( F_{s-1}^0 - k_s )\Phi(\pm \kappa_s) + \textup{IV}_s \sqrt{s}\phi(\pm \kappa_s)
     \Big)
\end{equation}
where $\Phi$ and $\phi$ are the normal cumulative distribution and density functions, respectively.  
Further,
\[
  \kappa_s  
  := 
  \frac{ F_{s-1}^0 - k_s }
   { \textup{IV}_s\sqrt{s} } 
\]
where $F_{s-1}^0 = P(0,s-1) - P(0,s)$ is the forward rate prevailing at time $0$, 
the strike is given by 
\begin{equation}
    \label{e:strike}
    k_s
    := 
    - P(0,s)^{-1}\frac{l_{s-1}^d-l_s^d}{l_{s-1}^h}\frac{UG_0}{(1+\vartheta)LP_0}
    + \frac{(1-\sigma)\rho_t-\gamma_t}{1+\vartheta}
\end{equation} \label{ref: k_s}
and $\textup{IV}_s$\label{ref: IV_s} is the caplet implied volatility known from market data.
\revise{Using the normal model at this point is, of course, an additional model choice.}

Therefore, estimate~\eqref{e:III_2} may be reformulated as 
\begin{equation}
    \label{e:III_3}
    \widehat{III}_{lb} \le
    III \le \widehat{III}_{ub}
\end{equation}
with 
\begin{align}
    \label{e:III_hat}
    \widehat{III}_{lb} 
    &:= 
    (1-gph)\Big( 
   F_0(1+F_0)^{-1} SF_0 +
    \vartheta\sum_{t=1}^{T-1}(P(0,t)-P(0,t+1) )l_{t-1}^h LP_0
   \Big) \\ 
   \notag
   &\phantom{===}
   +
  gph(1-gph)\sum_{t=1}^{T-1}
  \Big(1 - CV_{0,t}^1\, CV_t^2 \Big)
      \Big(1 - P(t,t+1)\Big)
      \mathcal{O}_t^+ (1+\vartheta)l_{t-1}^h LP_0
    \\ 
    \notag 
    \widehat{III}_{ub}
    &:= 
    (1-gph) (1-P(0,T)) SF_0 \\
    \notag 
    &\phantom{===}
    + (1-gph)gph \sum_{t=1}^{T-1} 
    \Big(1 + CV_{0,t}^1\, CV_t^2 \Big)
    \Big( 1 - P(t, t+1) \Big)
    \mathcal{O}_t^+ (1+\vartheta)l_{t-1}^h LP_0  \\
    \notag 
    &\phantom{===}
    + (1-gph)gph \sum_{t=2}^{T-1} \sum_{s=1}^{t-1} 
    \Big( 
       1 - \nu (1 - l_{t-s}^h )
      \Big)
      \Big(1 + CV_{s,t}^1\, CV_s^2 \Big)
      \Big(P(s,t) - P(s,t+1)\Big) \mathcal{O}_s^+ (1+\vartheta)l_{s-1}^h LP_0
\end{align}

\subsection{Approximating $COG$}
The shareholder cost of guarantees are defined in \eqref{e:COG}. We use again the simplified model \eqref{e:gshat} to estimate $COG$ by 
\begin{align}
    \label{e:COGhat}
    \widehat{COG}
    := 
    E\Big[
     \sum_{t=1}^T B_t^{-1} \widehat{gs}_t^-
    \Big]
    = 
    \sum_{t=1}^T \mathcal{O}_t^- (1+\vartheta)l_{t-1}^h LP_0 
\end{align}
where $\mathcal{O}_t^-$ is the value of the floorlet given in \eqref{e:Black}.

\subsection{Estimating $FDB$}
Under the assumptions of Section~\ref{sec:assump},
 the terms $COG$, $I$, $II$ and $III$ 
 can be estimated by  $\widehat{COG}$, $\widehat{I}=0$, $II\le\widehat{II}$ and $\widehat{III}_{lb} \le III \le \widehat{III}_{ub}$ as defined by \eqref{e:COGhat}, \eqref{e:Ihat}, \eqref{e:II_hat} and \eqref{e:III_hat}, respectively.
These estimates yield a lower bound, $\widehat{LB}$, and an upper bound, $\widehat{UB}$, for $FDB$: 
\begin{equation}
    \label{e:est-int}
    \widehat{LB} 
    \le FDB \le \widehat{UB}
\end{equation}
where 
\begin{align}
    \label{e:LB}
    \widehat{LB}
    &:= 
    SF_0 + gph\Big(LP_0+UG_0-GB\Big) 
    - \widehat{II}
    - \widehat{III}_{ub} 
    \\
    \label{e:UB}
    \widehat{UB}
    &:=
    SF_0 + gph\Big(LP_0+UG_0-GB\Big) + gph\cdot\widehat{COG}
    - \widehat{III}_{lb}. 
\end{align}
If the difference $\widehat{UB}-\widehat{LB}$ is sufficiently small (e.g., in comparison to $BV_0$), then $\widehat{FDB} = (\widehat{LB}+\widehat{UB})/2$ may be used as an estimator for $FDB$. 

These estimation formulas are analytic in the sense that they do not depend on a numerical model. For the reader's convenience we provide a compact list of data which have to be known or estimated in order to  calculate these bounds in Table~\ref{tbl:data_hat}.

\begin{remark}\label{rem:Art91}
The estimations $\widehat{LB}$ and $\widehat{UB}$ regard the $FDB$ as calculated with a stochastic cash flow model. The stochastic cash flow model is the numerical model used to generate cash flows relevant for best estimate calculation as in \cite{Gerstner08,VELP17}.
In those EU member states that have authorised Article~91(2) of Directive~2009/138/EC~\cite{L1}, the 
surplus fund (in fact, the part that is not used to compensate losses) is not considered as a liability and is therefore not part of the life assurance provision. Thus, if a company chooses to deduct the part, denoted by $SF_0^{\textup{Art.\,91}}$, of the surplus fund that is not used to absorb losses (in the risk neutral average over all scenarios), this is subtracted from the $FDB$ to yield the future discretionary benefits as reported to the supervisor and in financial statements, $FDB^{\textup{Art.\,91}} = FDB - SF_0^{\textup{Art.\,91}}$. Hence this information has to be known and when relevant the corresponding quantity has to be subtracted from the bounds $\widehat{LB}$ and $\widehat{UB}$. If this is the case we choose to approximate $SF_0^{\textup{Art.\,91}}$ by $SF_0$ itself in order to remain model free and thus use  $\widehat{LB}' = \widehat{LB} - SF_0$ and $\widehat{UB}' = \widehat{UB} - SF_0$. 
\end{remark}

\begin{table}[H] \centering{
\begin{enumerate}
\item
the balance sheet items $SF_0$, $LP_0$, $UG_0$, $GB$; 
\item 
the gross policyholder participation factor $gph$;
    \item 
    the initial discount curve $P(0,t)$ and interest rate implied volatilities $IV_t$;  
    \item 
    the coefficients of variation $CV_{s,t}^1$ and $CV_s^2$; 
    \item 
    duration factor $d$ \revise{in years};
    \item
    liability half life $h$ \revise{in years};
    \item 
    surplus fund fraction $\vartheta$;
    \item 
    bonus account factor $\sigma$;
    \item 
    bonus declaration bound $\nu$;
    \item 
    expected technical interest rate $\rho_t$;
    \item 
    expected technical gains rate $\gamma_t$; 
    \item
    the information concerning the application of Article 91 as in Remark~\ref{rem:Art91};
    \item 
    the projection time $T$;
\end{enumerate}
}
\caption{List of data needed to calculate $\widehat{LB}$ and $\widehat{UB}$.}
\label{tbl:data_hat}
\end{table}

\section{Public data} \label{sec:data}

\subsection{Allianz~Lebensversicherungs-AG: publicly reported values} 
The data in table~\ref{table:values} is taken from publicly available reports for the accounting years 2017-2019. The relevant references are listed in table~\ref{table: references}.

\begin{table}[H]
\centering
\begin{tabular}{lrrr}
Quantity                      & \multicolumn{1}{l}{2017} & \multicolumn{1}{c}{2018} & \multicolumn{1}{c}{2019}  \\ 
\hline
$L_0$                      & 189.8                                & 201.2                    & 219.6                     \\
$UG_0$                      & 41.4                                 & 32.8                     & 54.0                      \\
$SF_0$                      & 10.4                                 & 11.0                     & 11.5                      \\
Solvency~II value of~$SF_0$ & 10.9                                 & 10.5                     & 11.3                      \\
$GB$                        & 154.1                                & 158.8                    & 195.2                     \\
$FDB$                       & 48.6                                 & 46.2                     & 47.4                      \\
\hline
\end{tabular}
\vspace{2ex} \caption{Allianz Lebensversicherungs-AG: public data for 2017-2019, values are in billion euros.}
\label{table:values}
\end{table}

The value of $UG_0$ is already scaled to $L_0$, which is in line with the general assumption $(\ref{e:BV = LP + SF})$. The reason behind this scaling is that according to \cite[§~3]{MindZV} only the fraction of the capital gains, corresponding to the assets scaled to cover the average value of liabilities in the accounting year under consideration, contribute to the gross surplus.

As for $L_0$, we adjust the local GAAP value of life insurance with profit participation for necessary regrouping of business, as explained in \cite[p.~52]{SFCR}, \cite[p.~46]{SFCR2018}, \cite[p.~46]{SFCR2019} for the different accounting years 2017--2019. 

\begin{table}[H]
\begin{minipage}{\textwidth}
\centering
\begin{tabular}{lrrr}
Quantity                      & Source for 2017                                              & Source for 2018  & Source for 2019  \\ 
\hline
$L_0$\footnote{Versicherung mit Überschussbeteiligung}                      & \cite[p.~46,~52]{SFCR} & \cite[p.~42,~46]{SFCR2018}   & \cite[p.~42,~46]{SFCR2019}                                                                                           \\
$UG_0$\footnote{Stille Reserven der einzubeziehenden Kapitalanlagen}                      & \cite[p.~46]{GB}   & \cite[p.~42]{GB2018} & \cite[p.~46]{GB2019}                                                                                \\
$SF_0$\footnote{Rückstellung für Beitragsrückerstattung abzüglich\\festgelegte, aber noch nicht zugeteilte Teile}                     & \cite[p.~55]{GB}   &  \cite[p.~51]{GB2018} &  \cite[p.~55]{GB2019}\\
Solvency II value of~$SF_0$\footnote{Überschussfonds} & \cite[p.~52]{SFCR} & \cite[p.~46]{SFCR2018}  & \cite[p.~46]{SFCR2019}                                                                                                                  \\
$GB$\footnote{Bester Schätzwert: Wert für garantierte Leistungen}                        & \cite[p.~46]{SFCR} & \cite[p.~42]{SFCR2018}  & \cite[p.~42]{SFCR2019}                                                                                \\
$FDB$\footnote{Bester Schätzwert: zukünftige Überschussbeteiligung}                       & \cite[p.~46]{SFCR} &  \cite[p.~42]{SFCR2018} &  \cite[p.~42]{SFCR2019}                                                                           \\
\hline
\end{tabular}\vspace{2ex} \caption{Allianz Lebensversicherungs-AG: references for the data listed in table~\ref{table:values}.}
\label{table: references}
\end{minipage}
\end{table}

\subsection{Estimating technical gains from market data}\label{sec:tech_gains}
For the German life insurance market, technical gains can be determined from tables~130 and 141 in \cite{BaFin}. The relevant items are stated below:
\begin{table}[H]
\begin{minipage}{\textwidth}
\centering
\begin{tabular}{p{0.5\textwidth}crrr}
Quantity                      & \multicolumn{1}{l}{Symbol} & \multicolumn{1}{l}{2017} & \multicolumn{1}{c}{2018} & \multicolumn{1}{c}{2019}  \\ 
\hline
Gross surplus net of direct policyholder declarations\footnote{\"{U}berschuss} & $a$                       & 8.3                               & 9.9                    & 11.3                     \\
Direct policyholder declarations\footnote{Direktgutschrift} & $b$                     & 2.3                                 & 2.1                     & 2.1                      \\
Share of gross surplus allocated to the surplus fund\footnote{Zuf\"{u}hrung zur RfB} & $c$                     & 6.4                                 & 8.1                     & 9.3                      \\
Interest margin\footnote{Kapitalanlagenergebnis 1~b)} & $d$ & 3.5                                 & 5.2                     & 6.1                      \\
Gross technical provisions for direct business\footnote{Tabelle 130, Versicherungstechnische R\"{u}ckstellungen brutto, selbst abgeschlossenes Gesch\"{a}ft} & $e$ & 991.4
 & 1011.1
 & 1069.1\\
Gross technical provisions of those contracts where the investment risk is carried by the policyholder\footnote{Bilanzposten a)~6~a) brutto: Tabelle 130, Versicherungstechnische R\"{u}ckstellungen, soweit das Anlagerisiko vom Versicherungsnehmer getragen wird} & $f$                    & 109.1                                 & 101.7                     & 124.8                      \\
\hline
\end{tabular}
\vspace{2ex} \caption{BaFin: public data for 2017-2019, values are in billion euros.}
\label{table: items for gamma}
\end{minipage}
\end{table}

We estimate the technical gains relative to the life assurance provisions by $\widehat{\gamma}_{LP} := (a+b-d)/(e-f)$ and find
\begin{table}[H]
\centering
\begin{tabular}{lrrr}
                      & \multicolumn{1}{l}{2017} & \multicolumn{1}{c}{2018} & \multicolumn{1}{c}{2019}  \\ 
\hline
$\widehat{\gamma}_{LP}$                      & 0.80\,\%                               & 0.74\,\%                    & 0.78\,\%                     \\
\hline
\end{tabular}
\vspace{2ex} \caption{Values of $\gamma$ for 2017-2019.}
\label{table:gamma}
\end{table}

\subsection{Estimating $\rho$}\label{sec:rho}
The average technical interest rate of the Allianz Lebensversicherungs-AG can be derived from the distribution of life assurance provision over the guaranteed interest rates, with the following results:
\begin{table}[H]
\centering
\begin{tabular}{lrrr}
                      & \multicolumn{1}{l}{2017} & \multicolumn{1}{c}{2018} & \multicolumn{1}{c}{2019}  \\ 
\hline
$\widehat{\rho}$                      & 2.63\,\%                               & 2.52\,\%                    & 2.38\,\%                     \\
\hline
\end{tabular}
\vspace{2ex} \caption{Values of $\widehat{\rho}$ for 2017-2019.}
\label{table: values of rho}
\end{table}
The underlying data can be found in \cite[p.~34]{GB}, \cite[p.~33]{GB2018}, and \cite[p.~37]{GB2019}. (To obtain the values stated in table~\ref{table: values of rho}, we have taken the upper end points where technical interest rate intervals are provided.)


\subsection{Calculating $gph$}\label{sec:gph}
The net policyholder shares, $nph$ for the accounting years 2017-2019 can be obtained via $nph=(b+c)/a$ from the values collected in the table below (see \cite[p.~9]{GB}, \cite[p.~9]{GB2018}, and \cite[p.~8]{GB2019} for accounting years 2017-2019).
\begin{table}[H]
\begin{minipage}{\textwidth}
\centering
\begin{tabular}{lcrrr}
        Quantity              & Symbol & \multicolumn{1}{l}{2017} & \multicolumn{1}{c}{2018} & \multicolumn{1}{c}{2019}  \\ 
\hline
Gross surplus net of direct policyholder declarations\footnote{Brutto\"{u}berschuss} & $a$                     & 2.6                              & 3.1                   & 3.6  \\
Share of gross surplus allocated to the surplus fund\footnote{Zuf\"{u}hrung zur RfB} & $b$ & 2.0 & 2.3 & 2.9 \\
Direct policyholder declarations\footnote{Direktgutschrift} & $c$ & 0.1 & 0.1 & 0.2 \\
\hline
\end{tabular}
\vspace{2ex} \caption{Values are in billion euros.}
\label{table: items for gph}
\end{minipage}
\end{table}

The gross policyholder share, $gph$ is calculated from $nph$ according to the relation $gph = (1-\tau) nph / (1-\tau\cdot nph)$. Applying the German tax rate of $\tau = 29.9\,\%$ \cite[p.~16]{BMFtax}
yields the following table:
\begin{table}[H]
\centering
\begin{tabular}{lrrr}
        Quantity              & \multicolumn{1}{l}{2017} & \multicolumn{1}{c}{2018} & \multicolumn{1}{c}{2019}  \\ 
\hline
$nph$ & 80.8\,\% & 77.9\,\% & 85.6\,\% \\
$gph$ & 74.7\,\% & 71.2\,\% & 80.6\,\% \\
\hline
\end{tabular}
\vspace{2ex} \caption{Values of $gph$ for 2017-2019.}
\label{table:gph}
\end{table}

For the estimation of Term~$III$ we fix $gph = 75.5\,\%$ as the average of the values in table~\ref{table:gph}.

\subsection{Discount rates}
The following are the publicly available EIOPA discount rates for $2017$, $2018$ and $2019$. 

\begin{table}[H]
\centering
\begin{tabular}{llllllllllll}
$t$ & $P_{0,t}$ & $t$ & $P_{0,t}$                                       & $t$ & $P_{0,t}$ & $t$ & $P_{0,t}$ & $t$ & $P_{0,t}$ & $t$ & $P_{0,t}$  \\ 
\hline
1   &    1.003  & 11  & \begin{tabular}[c]{@{}l@{}}0.902\\\end{tabular} & 21  &    0.740  & 31  &    0.534  & 41  &    0.362  & 51  &    0.241   \\
2   & 1.004     & 12  & 0.885                                           & 22  & 0.720~    & 32  & 0.514~    & 42  & 0.348~    & 52  & 0.232~     \\
3   & 1.001     & 13  & 0.868                                           & 23  & 0.700     & 33  & 0.495     & 43  & 0.334     & 53  & 0.222~     \\
4   & 0.996     & 14  & 0.850                                           & 24  & 0.679     & 34  & 0.477     & 44  & 0.321~    & 54  & 0.214~     \\
5   & 0.988     & 15  & 0.834                                           & 25  & 0.658~    & 35  & 0.459~    & 45  & 0.308~    & 55  & 0.205~     \\
6   & 0.977     & 16  & 0.819                                           & 26  & 0.637     & 36  & 0.441~    & 46  & 0.296~    & 56  & 0.197~     \\
7   & 0.965     & 17  & 0.804                                           & 27  & 0.616~    & 37  & 0.424~    & 47  & 0.284~    & 57  & 0.189~     \\
8   & 0.951     & 18  & 0.790                                           & 28  & 0.595~    & 38  & 0.408     & 48  & 0.273~    & 58  & 0.181~     \\
9   & 0.936     & 19  & 0.774                                           & 29  & 0.574~    & 39  & 0.392~    & 49  & 0.262~    & 59  & 0.174~     \\
10  & 0.920     & 20  & 0.758                                           & 30  & 0.554     & 40  & 0.377     & 50  & 0.252     & 60  & 0.167      \\
\hline
\end{tabular}\vspace{2ex} \caption{Euro discount rates as of 31.12.2017. The rates are with volatility adjustment. 
Source: \cite{RFR}}
\label{table: Euro 2017}
\end{table}

\begin{table}[H]
\centering
\begin{tabular}{llllllllllll}
$t$ & $P_{0,t}$ & $t$ & $P_{0,t}$ & $t$ & $P_{0,t}$ & $t$ & $P_{0,t}$ & $t$ & $P_{0,t}$ & $t$ & $P_{0,t}$  \\ 
\hline
1   &    1.001  & 11  &    0.890  & 21  &    0.722  & 31  &    0.525  & 41  &    0.360  & 51  &    0.244   \\
2   & 1.001     & 12  & 0.872~    & 22  & 0.703~    & 32  & 0.506~    & 42  & 0.347~    & 52  & 0.234~     \\
3   & 0.998     & 13  & 0.853~    & 23  & 0.684~    & 33  & 0.488~    & 43  & 0.334~    & 53  & 0.225~     \\
4   & 0.992     & 14  & 0.835~    & 24  & 0.664~    & 34  & 0.470~    & 44  & 0.321~    & 54  & 0.217~     \\
5   & 0.983     & 15  & 0.818~    & 25  & 0.644~    & 35  & 0.453~    & 45  & 0.309~    & 55  & 0.208~     \\
6   & 0.972     & 16  & 0.803~    & 26  & 0.623~    & 36  & 0.436~    & 46  & 0.297~    & 56  & 0.200~     \\
7   & 0.958     & 17  & 0.788~    & 27  & 0.603~    & 37  & 0.420~    & 47  & 0.285~    & 57  & 0.192~     \\
8   & 0.943     & 18  & 0.773~    & 28  & 0.583~    & 38  & 0.405~    & 48  & 0.274~    & 58  & 0.185~     \\
9   & 0.926     & 19  & 0.757~    & 29  & 0.563~    & 39  & 0.389~    & 49  & 0.264~    & 59  & 0.178~     \\
10  & 0.908     & 20  & 0.740     & 30  & 0.544     & 40  & 0.375     & 50  & 0.254     & 60  & 0.171      \\
\hline
\end{tabular}\vspace{2ex} \caption{Euro discount rates as of 31.12.2018. The rates are with volatility adjustment.  Source: \cite{RFR}}
\label{table: Euro 2018}
\end{table}

\begin{table}[H]
\centering
\begin{tabular}{llllllllllll}
$t$ & $P_{0,t}$ & $t$ & $P_{0,t}$ & $t$ & $P_{0,t}$ & $t$ & $P_{0,t}$ & $t$ & $P_{0,t}$ & $t$ & $P_{0,t}$  \\ 
\hline
1   &    1.004  & 11  &    0.975  & 21  &    0.878  & 31  &    0.665  & 41  &    0.466  & 51  &    0.320   \\
2   & 1.006~    & 12  & 0.967~    & 22  & 0.861~    & 32  & 0.643~    & 42  & 0.449~    & 52  & 0.308~     \\
3   & 1.008    & 13  & 0.957~    & 23  & 0.842~    & 33  & 0.621~    & 43  & 0.432~    & 53  & 0.296~     \\
4   & 1.009~    & 14  & 0.947~    & 24  & 0.821~    & 34  & 0.600~    & 44  & 0.416~    & 54  & 0.285~     \\
5   & 1.008~    & 15  & 0.937~    & 25  & 0.800~    & 35  & 0.579~    & 45  & 0.401~    & 55  & 0.275~     \\
6   & 1.006~    & 16  & 0.929~    & 26  & 0.778~    & 36  & 0.559~    & 46  & 0.386~    & 56  & 0.264~     \\
7   & 1.001~    & 17  & 0.922~    & 27  & 0.755~    & 37  & 0.539~    & 47  & 0.372~    & 57  & 0.254~     \\
8   & 0.996     & 18  & 0.914~    & 28  & 0.733~    & 38  & 0.520~    & 48  & 0.358~    & 58  & 0.245~     \\
9   & 0.990~    & 19  & 0.904~    & 29  & 0.710~    & 39  & 0.501~    & 49  & 0.345~    & 59  & 0.236~     \\
10  & 0.982     & 20  & 0.893     & 30  & 0.687     & 40  & 0.483     & 50  & 0.332     & 60  & 0.227      \\
\hline
\end{tabular}\vspace{2ex} \caption{Euro discount rates as of 31.12.2019. The rates are with volatility adjustment. Source: \cite{RFR}}
\label{table: Euro 2019}
\end{table}

\section{Estimation of $\widehat{FDB}$ from public data}\label{sec:est}

We use the publicly available data  collected in Section~\ref{sec:data} to find the estimation interval for $\widehat{FDB}$ according to \eqref{e:est-int}, and compare the result with numerically calculated $FDB$ contained in the public data.\footnote{The calculations have been carried out in $R$ and the script files can be provided upon request.} In order to calculate $\widehat{LB}$ and $\widehat{UB}$ we have to know or estimate the  data listed in Table~\ref{tbl:data_hat}. This is done as follows: 

\begin{enumerate}
\item
the balance sheet items $SF_0$, $LP_0$, $UG_0$, $GB$: These are given in Table~\ref{table:values} with $L_0 - SF_0 = BV_0 - SF_0 = LP_0$. 
\item 
the gross policyholder participation factor $gph$: We use the average value of $gph = 75.5\,\%$ given in Table~\ref{table:gph}. The average is employed since this factor subsequently remains constant over the full projection time and should not depend on special circumstances at valuation time.
    \item 
    the initial discount curve $P(0,t)$ and interest rate implied volatilities $IV_t$: The discount curve is the relevant EIOPA curve as listed in Tables~\ref{table: Euro 2017}, \ref{table: Euro 2018} and \ref{table: Euro 2019}; the implied volatilities are taken from Bloomberg (end of year 2019) for available maturities, linearly interpolated between available maturities and extrapolated by keeping the last available volatility constant. This leads to $IV_t = 10 + 50(t-1)/21$ for $1\le t \le 21$ and $IV_t=50$ for $t\ge21$, expressed in basis points.  Additional sensitivity analysis is performed.
    \item 
    the coefficients of variation $CV_{s,t}^1$ and $CV_s^2$: Since the product of two such coefficients is expected to be small, and it is the product that enters the calculation of $\widehat{LB}$ and $\widehat{UB}$, these are estimated as $CV_{s,t}^1 \revise{\cdot} CV_s^2 = 0$. While this parameter choice is certainly very practical it is not very well founded from a theoretical perspective, and an estimation from historical data would be a more justifiable approach. 
    \item 
    duration factor $d$: This factor is known only to the company under consideration. We set $d=8$, and perform sensitivity analysis on this assumption. 
    \item
    liability half life $h$: This factor is known only to the company under consideration. We set $h=10$, and perform sensitivity analysis on this assumption. 
    \item 
    surplus fund fraction $\vartheta$: We take $\vartheta = SF_0/LP_0$, and perform sensitivity analysis.
    \item 
    bonus account factor $\sigma$: We choose $\sigma = 20\,\%$, and perform sensitivity analysis. 
    \item 
    bonus declaration lower bound $\nu$: We choose $\nu = 75\,\%$, and perform sensitivity analysis. 
    \item 
    expected technical interest rate $\rho_t$: We use the values constant $\hat{\rho} = \rho_t$ as listed in Table~\ref{table: values of rho}, and perform sensitivity analysis. 
    \item 
    expected technical gains rate $\gamma_t$: We use the constant  values $\hat{\gamma} = \gamma_t$ as listed in Table~\ref{table:gamma}, and perform sensitivity analysis. 
    \item 
    the information concerning the application of Article 91 as in Remark~\ref{rem:Art91}: the company in question does apply Article~91 as stated in \cite{SFCR} and  seen in Table~\ref{table:values}. Hence we subtract $SF_0$ from the bounds to estimate $FDB$.
    \item 
    the projection time $T$: we have chosen $T=50$ years to reflect the long term nature of life insurance;
\end{enumerate}

\revise{
\begin{remark}\label{rem:CV}
The coefficient of variation $CV_s^2$ cannot be estimated from market data. To estimate it correctly one would need a full numerical model to calculate the variation of $gs_s^+$. To avoid the need for a numerical best estimate model one may also approximate $CV_s^2$ by $CV[B_s^{-1}\widehat{gs}_s]$ using the simplified model \eqref{e:gshat}. However, management rules are usually structured so as to reduce variations in bonus declarations. Hence we expect that $CV_s^2$ should be sufficiently small so that the product $CV_{s,t}^1 \revise{\cdot} CV_s^2$ is negligible. 
\end{remark}
}

The results corresponding to these assumptions are referred to as the base case, and shown in Tables~\ref{table:Nr 0 } and \ref{table:Nr 1 } in absolute value (billion Euros) and in percent of the initial market value $MV_0 = LP_0+SF_0+UG_0$, respectively. The value $FDB$ is the numerically calculated number as reported by the company, see Table~\ref{table:values}. 

We use throughout the notation $\delta = \widehat{FDB}-FDB$ and  $\epsilon = (UB-LB)/2$, either as absolute values or relative to $MV_0$, as indicated. The estimation interval is thus given by $\widehat{FDB}\pm\epsilon$ and the estimation is considered successful if $|\delta|<\epsilon$ such that the true value, $FDB$, lies within this interval. This holds for the base case as-well as for all sensitivities.

\begin{table}[H]\phantom{X}\[
\revise{
\begin{matrix}
  & LP_0 & SF_0 & UG_0 & GB & FDB & \widehat{FDB} & \widehat{LB} & \widehat{UB} & \epsilon & \delta & \widehat{II} & \widehat{COG} \\ 
  2017 & 179.40 & 10.40 & 41.40 & 154.10 & 48.60 & 46.78 & 43.82 & 49.73 & 2.96 & -1.82 & 1.10 & 0.50 \\ 
  2018 & 190.20 & 11.00 & 32.80 & 158.80 & 46.20 & 45.16 & 42.24 & 48.08 & 2.92 & -1.04 & 1.07 & 0.85 \\ 
  2019 & 208.10 & 11.50 & 54.00 & 195.20 & 47.40 & 47.44 & 44.05 & 50.84 & 3.40 & 0.04 & 1.39 & 1.50 \\ 
   \end{matrix} }
\]\vspace{0.8mm}
\caption{Base case, displayed numbers are in billion Euros;} \label{table:Nr 0 } \end{table}

\begin{table}[H]\phantom{X}\[
\revise{
\begin{matrix}
  & LP_0 & SF_0 & UG_0 & GB & FDB & \widehat{FDB} & \widehat{LB} & \widehat{UB} & \epsilon & \delta & \widehat{II} & \widehat{COG} \\ 
 2017 & 77.60 & 4.50 & 17.91 & 66.65 & 21.02 & 20.23 & 18.95 & 21.51 & 1.28 & -0.79 & 0.48 & 0.21 \\ 
  2018 & 81.28 & 4.70 & 14.02 & 67.86 & 19.74 & 19.30 & 18.05 & 20.55 & 1.25 & -0.44 & 0.46 & 0.36 \\ 
  2019 & 76.06 & 4.20 & 19.74 & 71.35 & 17.32 & 17.34 & 16.10 & 18.58 & 1.24 & 0.02 & 0.51 & 0.55 \\ 
   \end{matrix} }
\]\vspace{0.8mm}
\caption{Base case, displayed numbers are in percent of $MV_0 = LP_0+SF_0+UG_0$;} \label{table:Nr 1 } \end{table}

Table~\ref{table:Nr 1 } shows that the estimation error, $\delta$, compared to the true value is in all three cases below $1\,\%$ of the initial market value, $MV_0$. We view this as quite a remarkable result for an analytically calculated approximation. Further, it is shown that the influence of $\widehat{II}$ on $\widehat{LB}$ is quite small. The estimated cost of guarantee, $\widehat{COG}$, \revise{increases noticeably} from 2018 to 2019. This is due to the significantly lower interest rate curve, as can be seen by comparing Tables~\ref{table: Euro 2017}, \ref{table: Euro 2018} and \ref{table: Euro 2019}. Generally speaking, one may also remark that the ratio $SF_0/LP_0$ has a significant impact on cost of guarantees: if $SF_0$ is large compared to $LP_0$ the basis, $BV_0=SF_0+LP_0$ for the return on assets is comparatively large and this is advantageous from the company's point of view since the guaranteed interest rate acts only on $V_0\le LP_0$.

\begin{table}[H]\phantom{X}\[
\revise{
\begin{matrix}
  & FDB & \widehat{FDB} & \widehat{LB} & \widehat{UB} & \epsilon & \delta \\ 
  2017 & 21.02 & 20.16 & 18.97 & 21.35 & 1.19 & -0.86 \\ 
  2018 & 19.74 & 19.18 & 18.08 & 20.29 & 1.10 & -0.56 \\ 
  2019 & 17.32 & 17.16 & 16.13 & 18.19 & 1.03 & -0.16 \\ 
   \end{matrix} }
\quad\left|\quad \revise{ \begin{matrix}
 \Delta^{\textup{rel}}\,\widehat{FDB} & \Delta^{\textup{rel}}\,\widehat{LB} & \Delta^{\textup{rel}}\,\widehat{UB} \\ 
 -0.35 & 0.08 & -0.74 \\ 
  -0.61 & 0.13 & -1.32 \\ 
  -1.01 & 0.19 & -2.26 \\ 
   \end{matrix} }
\right.\]\vspace{0.8mm}
\caption{Sensitivity: volatility is reduced by $50\,\%$, $IV_t' = IV_t/2$. 
LHS: displayed numbers are in percent of $MV_0 = LP_0+SF_0+UG_0$; 
RHS: difference to base case in percent of $FDB$;} 
\label{table:Nr 2 } \end{table}

Table~\ref{table:Nr 2 } shows the corresponding results for an implied volatility curve which has been reduced by $50\,\%$. The effect is most pronounced on the estimation of the upper bound, \revise{$\widehat{UB}$}, since it leads to a reduced cost of guarantees. Further, we notice that the estimation interval $\widehat{FDB}\pm\epsilon$ shrinks quite strongly. This makes sense since the estimation of a stochastic quantity should improve as the underlying volatility is reduced.  

The notation $\Delta^{\textup{rel}}$ in Table~\ref{table:Nr 2 }, as-well as below, is understood relative to the base case and in percent of $FDB$, that is $\Delta^{\textup{rel}}X = 100(X-X_0)/FDB$ where $X$ and $X_0$ are in billion Euros, and $X_0$ is taken from Table~\ref{table:Nr 0 }. Thus all the sensitivities in Table~\ref{table:Nr 2 } are at most of the order of $5\,\%\,FDB$, and would be almost $0\,\%$ when compared to $MV_0$. This observation holds also for all the other sensitivities considered subsequently.

\begin{table}[H]\phantom{X}\[
\revise{
\begin{matrix}
  & FDB & \widehat{FDB} & \widehat{LB} & \widehat{UB} & \epsilon & \delta \\ 
 2017 & 21.02 & 20.47 & 18.90 & 22.05 & 1.57 & -0.55 \\ 
  2018 & 19.74 & 19.62 & 17.99 & 21.26 & 1.64 & -0.12 \\ 
  2019 & 17.32 & 17.75 & 16.02 & 19.48 & 1.73 & 0.43 \\ 
   \end{matrix} }
\quad\left|\quad \revise{ \begin{matrix}
  \Delta^{\textup{rel}}\,\widehat{FDB} & \Delta^{\textup{rel}}\,\widehat{LB} & \Delta^{\textup{rel}}\,\widehat{UB} \\ 
 1.15 & -0.25 & 2.55 \\ 
  1.62 & -0.32 & 3.59 \\ 
  2.38 & -0.44 & 5.17 \\ 
   \end{matrix} }
\right.\]\vspace{0.8mm}
\caption{ 
\revise{Sensitivity:} volatility increased by $50\,\%$, $IV_t' = 1.5\,IV_t$; LHS: displayed numbers are in percent of $MV_0$; 
              RHS: difference to base case in percent of $FDB$;} 
           \label{table:Nr 13 } \end{table}

Table~\ref{table:Nr 13 } shows the effect of increasing volatility by $50\,\%$ which leads to a noticeable increase in $\widehat{COG}$ and therefore of $\widehat{UB}$. Increasing the volatility thus also implies a larger estimation error $\pm\epsilon$. This effect is most pronounced for 2019 because of the very low interest rate environment.

\begin{table}[H]\phantom{X}\[
\revise{
\begin{matrix}
  & FDB & \widehat{FDB} & \widehat{LB} & \widehat{UB} & \epsilon & \delta \\ 
 2017 & 21.02 & 20.02 & 18.62 & 21.43 & 1.41 & -1.00 \\ 
  2018 & 19.74 & 19.07 & 17.72 & 20.41 & 1.34 & -0.68 \\ 
  2019 & 17.32 & 17.12 & 15.84 & 18.41 & 1.29 & -0.20 \\ 
   \end{matrix}  }
\quad\left|\quad \revise{ \begin{matrix}
 \Delta^{\textup{rel}}\,\widehat{FDB} & \Delta^{\textup{rel}}\,\widehat{LB} & \Delta^{\textup{rel}}\,\widehat{UB} \\ 
  -1.01 & -1.60 & -0.39 \\ 
  -1.19 & -1.67 & -0.69 \\ 
  -1.24 & -1.52 & -0.99 \\ 
   \end{matrix}   }
\right.\]\vspace{0.8mm} 
\caption{Sensitivity:  $\rho'=0.75\,\rho$;
LHS: displayed numbers are in percent of $MV_0 = LP_0+SF_0+UG_0$; 
RHS: difference to base case in percent of $FDB$;} \label{table:Nr 3 } \end{table}

Table~\ref{table:Nr 3 } shows that the effect of reducing the (constant) technical interest rate by $25\,\%$ is quite small compared to $FDB$.

\begin{table}[H]\phantom{X}\[
\revise{
\begin{matrix}
  & FDB & \widehat{FDB} & \widehat{LB} & \widehat{UB} & \epsilon & \delta \\ 
 2017 & 21.02 & 20.57 & 19.27 & 21.86 & 1.30 & -0.46 \\ 
  2018 & 19.74 & 19.82 & 18.33 & 21.30 & 1.49 & 0.07 \\ 
  2019 & 17.32 & 17.63 & 16.35 & 18.91 & 1.28 & 0.31 \\ 
   \end{matrix} }
\quad\left|\quad \revise{ \begin{matrix}
 \Delta^{\textup{rel}}\,\widehat{FDB} & \Delta^{\textup{rel}}\,\widehat{LB} & \Delta^{\textup{rel}}\,\widehat{UB} \\ 
 1.58 & 1.50 & 1.69 \\ 
  2.62 & 1.41 & 3.81 \\ 
  1.69 & 1.43 & 1.90 \\ 
   \end{matrix} }
\right.\]\vspace{0.8mm}
\caption{Sensitivity:  $\rho'=1.25\,\rho$;
LHS: displayed numbers are in percent of $MV_0 = LP_0+SF_0+UG_0$; 
RHS: difference to base case in percent of $FDB$;} 
\label{table:Nr 4 } \end{table}

Table~\ref{table:Nr 4 } shows that the effect of increasing the (constant) technical interest rate by $25\,\%$ is quite small compared to $FDB$. However, it can also be seen that this conclusion depends on the specific circumstances, as the conditions corresponding to 2018   lead to a \revise{slightly more pronounced} effect. 

\begin{table}[H]\phantom{X}\[
\revise{
\begin{matrix}
  & FDB & \widehat{FDB} & \widehat{LB} & \widehat{UB} & \epsilon & \delta \\ 
 2017 & 21.02 & 20.57 & 19.44 & 21.71 & 1.14 & -0.45 \\ 
  2018 & 19.74 & 19.76 & 18.49 & 21.03 & 1.27 & 0.02 \\ 
  2019 & 17.32 & 17.69 & 16.56 & 18.83 & 1.13 & 0.37 \\ 
   \end{matrix} }
\quad\left|\quad \revise{ \begin{matrix}
 \Delta^{\textup{rel}}\,\widehat{FDB} & \Delta^{\textup{rel}}\,\widehat{LB} & \Delta^{\textup{rel}}\,\widehat{UB} \\ 
 1.63 & 2.28 & 0.97 \\ 
  2.34 & 2.23 & 2.45 \\ 
  2.05 & 2.66 & 1.41 \\ 
   \end{matrix} }
\right.\]\vspace{0.8mm}
\caption{
Sensitivity:  $\gamma'=0.5\,\gamma$;
LHS: displayed numbers are in percent of $MV_0$; 
RHS: difference to base case in percent of $FDB$;
%
} 
\label{table:Nr 5 } \end{table}

Table~\ref{table:Nr 5 } shows that the effect of reducing the (constant) technical gains rate by $50\,\%$ is quite small compared to $FDB$. 

\begin{table}[H]\phantom{X}\[
\revise{
\begin{matrix}
  & FDB & \widehat{FDB} & \widehat{LB} & \widehat{UB} & \epsilon & \delta \\ 
 2017 & 21.02 & 19.95 & 18.46 & 21.45 & 1.49 & -1.07 \\ 
  2018 & 19.74 & 19.01 & 17.58 & 20.43 & 1.43 & -0.74 \\ 
  2019 & 17.32 & 17.03 & 15.63 & 18.44 & 1.40 & -0.29 \\ 
   \end{matrix} }
\quad\left|\quad \revise{ \begin{matrix}
 \Delta^{\textup{rel}}\,\widehat{FDB} & \Delta^{\textup{rel}}\,\widehat{LB} & \Delta^{\textup{rel}}\,\widehat{UB} \\ 
 -1.34 & -2.35 & -0.31 \\ 
  -1.47 & -2.38 & -0.56 \\ 
  -1.77 & -2.70 & -0.84 \\ 
   \end{matrix} } 
\right.\]\vspace{0.8mm}
\caption{
Sensitivity:  $\gamma'=1.5\,\gamma$;
LHS: displayed numbers are in percent of $MV_0$; 
RHS: difference to base case in percent of $FDB$;
%
} \label{table:Nr 6 } \end{table}

Table~\ref{table:Nr 6 } shows that the effect of increasing the (constant) technical gains rate by $50\,\%$ is quite small compared to $FDB$. 

\begin{table}[H]\phantom{X}\[
\revise{
\begin{matrix}
  & FDB & \widehat{FDB} & \widehat{LB} & \widehat{UB} & \epsilon & \delta \\ 
 2017 & 21.02 & 20.30 & 18.98 & 21.61 & 1.32 & -0.72 \\ 
  2018 & 19.74 & 19.37 & 18.08 & 20.66 & 1.29 & -0.37 \\ 
  2019 & 17.32 & 17.37 & 16.12 & 18.63 & 1.25 & 0.05 \\ 
   \end{matrix} }
\quad\left|\quad \revise{ \begin{matrix}
 \Delta^{\textup{rel}}\,\widehat{FDB} & \Delta^{\textup{rel}}\,\widehat{LB} & \Delta^{\textup{rel}}\,\widehat{UB} \\ 
 0.29 & 0.12 & 0.49 \\ 
  0.35 & 0.15 & 0.54 \\ 
  0.19 & 0.08 & 0.25 \\  
   \end{matrix} }
\right.\]\vspace{0.8mm}
\caption{Sensitivity:  $\theta'=0.5\,\theta$;
LHS: displayed numbers are in percent of $MV_0$; 
RHS: difference to base case in percent of $FDB$;
%
} \label{table:Nr 7 } \end{table}

Table~\ref{table:Nr 7 } shows that the effect of reducing  $\theta$, estimated by $\theta=SF_0/LP_0$, by $50\,\%$ is  quite small compared to $FDB$.

\begin{table}[H]\phantom{X}\[
\revise{
\begin{matrix}
  & FDB & \widehat{FDB} & \widehat{LB} & \widehat{UB} & \epsilon & \delta \\ 
 2017 & 21.02 & 20.17 & 18.93 & 21.41 & 1.24 & -0.85 \\ 
  2018 & 19.74 & 19.23 & 18.02 & 20.44 & 1.21 & -0.51 \\ 
  2019 & 17.32 & 17.31 & 16.08 & 18.54 & 1.23 & -0.01 \\ 
   \end{matrix} }
\quad\left|\quad \revise{ \begin{matrix}
 \Delta^{\textup{rel}}\,\widehat{FDB} & \Delta^{\textup{rel}}\,\widehat{LB} & \Delta^{\textup{rel}}\,\widehat{UB} \\ 
 -0.31 & -0.12 & -0.47 \\ 
  -0.35 & -0.15 & -0.52 \\ 
  -0.17 & -0.08 & -0.25 \\ 
   \end{matrix} }
\right.\]\vspace{0.8mm}
\caption{
Sensitivity:  $\theta'=1.5\,\theta$;
LHS: displayed numbers are in percent of $MV_0$; 
RHS: difference to base case in percent of $FDB$;
%
} \label{table:Nr 8 } \end{table}

Table~\ref{table:Nr 8 } shows that the effect of increasing  $\theta$, estimated by $\theta=SF_0/LP_0$, by $50\,\%$ is  quite small compared to $FDB$.

\begin{table}[H]\phantom{X}\[
\revise{
\begin{matrix}
  & FDB & \widehat{FDB} & \widehat{LB} & \widehat{UB} & \epsilon & \delta \\ 
 2017 & 21.02 & 20.48 & 19.36 & 21.60 & 1.12 & -0.54 \\ 
  2018 & 19.74 & 19.63 & 18.43 & 20.83 & 1.20 & -0.12 \\ 
  2019 & 17.32 & 17.60 & 16.48 & 18.71 & 1.12 & 0.27 \\ 
   \end{matrix} }
\quad\left|\quad \revise{ \begin{matrix}
 \Delta^{\textup{rel}}\,\widehat{FDB} & \Delta^{\textup{rel}}\,\widehat{LB} & \Delta^{\textup{rel}}\,\widehat{UB} \\ 
 1.15 & 1.91 & 0.43 \\ 
  1.67 & 1.90 & 1.43 \\ 
  1.50 & 2.19 & 0.76 \\ 
   \end{matrix} }
\right.\]\vspace{0.8mm}
\caption{
Sensitivity:  $\sigma'=0.5\,\sigma$;
LHS: displayed numbers are in percent of $MV_0$; 
RHS: difference to base case in percent of $FDB$;
} \label{table:Nr 9 } \end{table}

Table~\ref{table:Nr 9 } shows that the effect of reducing $\sigma$, estimated by the chosen value $\sigma=20\,\%$, by $50\,\%$ is quite small compared to $FDB$.

\begin{table}[H]\phantom{X}\[
\revise{
\begin{matrix}
  & FDB & \widehat{FDB} & \widehat{LB} & \widehat{UB} & \epsilon & \delta \\ 
 2017 & 21.02 & 20.01 & 18.55 & 21.47 & 1.46 & -1.01 \\ 
  2018 & 19.74 & 19.06 & 17.66 & 20.46 & 1.40 & -0.68 \\ 
  2019 & 17.32 & 17.10 & 15.72 & 18.49 & 1.38 & -0.22 \\ 
   \end{matrix} }
\quad\left|\quad \revise{ \begin{matrix}
 \Delta^{\textup{rel}}\,\widehat{FDB} & \Delta^{\textup{rel}}\,\widehat{LB} & \Delta^{\textup{rel}}\,\widehat{UB} \\ 
 -1.09 & -1.93 & -0.21 \\ 
  -1.21 & -1.97 & -0.45 \\ 
  -1.37 & -2.22 & -0.55 \\ 
   \end{matrix} }
\right.\]\vspace{0.8mm}
\caption{
Sensitivity:  $\sigma'=1.5\,\sigma$;
LHS: displayed numbers are in percent of $MV_0$; 
RHS: difference to base case in percent of $FDB$;
} \label{table:Nr 10 } \end{table}

Table~\ref{table:Nr 10 } shows that the effect of increasing $\sigma$, estimated by the chosen value $\sigma=20\,\%$, by $50\,\%$ is quite small compared to $FDB$.

\begin{table}[H]\phantom{X}\[
\revise{
\begin{matrix}
  & FDB & \widehat{FDB} & \widehat{LB} & \widehat{UB} & \epsilon & \delta \\ 
 2017 & 21.02 & 20.10 & 18.69 & 21.51 & 1.41 & -0.92 \\ 
  2018 & 19.74 & 19.19 & 17.83 & 20.55 & 1.36 & -0.55 \\ 
  2019 & 17.32 & 17.23 & 15.88 & 18.58 & 1.35 & -0.09 \\ 
   \end{matrix} } 
\quad\left|\quad \revise{ \begin{matrix}
 \Delta^{\textup{rel}}\,\widehat{FDB} & \Delta^{\textup{rel}}\,\widehat{LB} & \Delta^{\textup{rel}}\,\widehat{UB} \\ 
 -0.64 & -1.26 & 0.00 \\ 
  -0.56 & -1.13 & 0.00 \\ 
  -0.61 & -1.27 & 0.00 \\ 
   \end{matrix} }
\right.\]\vspace{0.8mm}
\caption{Sensitivity:  
        $\nu' = 0.75\,\nu$; LHS: displayed numbers are in percent of $MV_0$; 
              RHS: difference to base case in percent of $FDB$;} 
           \label{table:Nr 11 } \end{table}

Table~\ref{table:Nr 11 } shows that the effect of decreasing $\nu$, estimated by the chosen value $\nu=75\,\%$, by $25\,\%$ is quite small compared to $FDB$.

\begin{table}[H]\phantom{X}\[
\revise{ 
\begin{matrix}
  & FDB & \widehat{FDB} & \widehat{LB} & \widehat{UB} & \epsilon & \delta \\ 
 2017 & 21.02 & 20.36 & 19.22 & 21.51 & 1.15 & -0.66 \\ 
  2018 & 19.74 & 19.41 & 18.27 & 20.55 & 1.14 & -0.33 \\ 
  2019 & 17.32 & 17.45 & 16.32 & 18.58 & 1.13 & 0.13 \\ 
   \end{matrix} }
\quad\left|\quad \revise{ \begin{matrix}
 \Delta^{\textup{rel}}\,\widehat{FDB} & \Delta^{\textup{rel}}\,\widehat{LB} & \Delta^{\textup{rel}}\,\widehat{UB} \\ 
 0.62 & 1.26 & 0.00 \\ 
  0.56 & 1.13 & 0.00 \\ 
  0.63 & 1.27 & 0.00 \\ 
   \end{matrix} }
\right.\]\vspace{0.8mm}
\caption{Sensitivity:  
 $\nu' = 1.25\,\nu$; LHS: displayed numbers are in percent of $MV_0$; 
              RHS: difference to base case in percent of $FDB$;} 
           \label{table:Nr 12 } \end{table}

Table~\ref{table:Nr 12 } shows that the effect of increasing $\nu$, estimated by the chosen value $\nu=75\,\%$, by $25\,\%$ is quite small compared to $FDB$. \revise{The parameter $\nu$ does not enter the estimation formula \eqref{e:UB} for the upper bound and hence  $\Delta^{\textup{rel}}\,\widehat{UB} = 0$ in Tables~\ref{table:Nr 11 } and \ref{table:Nr 12 }.
}

\begin{table}[H]\phantom{X}\[
\revise{
\begin{matrix}
  & FDB & \widehat{FDB} & \widehat{LB} & \widehat{UB} & \epsilon & \delta \\ 
 2017 & 21.02 & 20.26 & 18.98 & 21.55 & 1.28 & -0.76 \\ 
  2018 & 19.74 & 19.38 & 18.06 & 20.71 & 1.32 & -0.36 \\ 
  2019 & 17.32 & 17.32 & 16.12 & 18.53 & 1.21 & -0.00 \\ 
   \end{matrix}}
\quad\left|\quad\revise{ \begin{matrix}
 \Delta^{\textup{rel}}\,\widehat{FDB} & \Delta^{\textup{rel}}\,\widehat{LB} & \Delta^{\textup{rel}}\,\widehat{UB} \\ 
 0.12 & 0.10 & 0.16 \\ 
  0.43 & 0.04 & 0.82 \\ 
  -0.11 & 0.08 & -0.32 \\ 
   \end{matrix}}
\right.\]\vspace{0.8mm}
\caption{Sensitivity:  
$d'=10=h$
LHS: displayed numbers are in percent of $MV_0$; 
              RHS: difference to base case in percent of $FDB$;} 
           \label{table:Nr 15 } \end{table}

Table~\ref{table:Nr 15 } shows that the effect of increasing $d=8$ to $d+2$, while leaving $h=10$ unchanged, is quite small compared to $FDB$.

\begin{table}[H]\phantom{X}\[
\revise{
\begin{matrix}
  & FDB & \widehat{FDB} & \widehat{LB} & \widehat{UB} & \epsilon & \delta \\ 
 2017 & 21.02 & 20.15 & 18.69 & 21.61 & 1.46 & -0.87 \\ 
  2018 & 19.74 & 19.24 & 17.80 & 20.68 & 1.44 & -0.51 \\ 
  2019 & 17.32 & 17.37 & 15.88 & 18.85 & 1.48 & 0.04 \\ 
   \end{matrix} }
\quad\left|\quad\revise{ \begin{matrix}
 \Delta^{\textup{rel}}\,\widehat{FDB} & \Delta^{\textup{rel}}\,\widehat{LB} & \Delta^{\textup{rel}}\,\widehat{UB} \\ 
 -0.39 & -1.23 & 0.47 \\ 
  -0.30 & -1.28 & 0.67 \\ 
  0.17 & -1.24 & 1.56 \\ 
   \end{matrix}}
\right.\]\vspace{0.8mm}
\caption{Sensitivity:  
$h' = 12 = h+2$;
LHS: displayed numbers are in percent of $MV_0$; 
RHS: difference to base case in percent of $FDB$;} 
           \label{table:Nr 14 } \end{table}
           
Table~\ref{table:Nr 14 } shows that the effect of increasing $h=10$ to $h+2$, while leaving $d=8$ unchanged, is quite small compared to $FDB$.


\section{Conclusions}\label{sec:concl}
The bounds \eqref{e:LB} and \eqref{e:UB} have been derived in a manner which is quite basic from the mathematical point of view but seems at the same time adequate for real world applications. 
We view the accuracy of $\widehat{FDB}$ as demonstrated by $\delta = (\widehat{FDB}-FDB)/MV_0 < 1\,\%$ and $\epsilon = (\widehat{UB}-\widehat{LB})/MV_0 < \revise{1.5}\,\%$ in all three cases in Table~\ref{table:Nr 1 } as quite remarkable. 

However, it would certainly also be interesting to further refine the model: for example one could relax the assumption~\ref{ass:ROA} and attempt to model $ROA_t$, or also the difference $ROA_t - E[ROA_t|\mathcal{F}_{t-1}]$ along the lines of \cite{Dorobantu_etal20}. Moreover, the estimation of the return on the deferred caplet in equation~\eqref{e:III_2} relies on simply estimating the absolute value of a correlation factor by $1$, and this estimate could possibly be improved.   At the same time the product, $CV_{s,t}^1CV_s^2$,  of the coefficients of variation as mentioned in Table~\ref{tbl:data_hat} has been simply set to $0$ in Section~\ref{sec:est}. The validity of this parameter choice remains to be analyzed.

\section*{Competing interests declaration} 
The authors declare no competing interests.

\revise{
\section*{Acknowledgements} 
We thank our colleagues Wolfgang Herold and Sanela Omerovic for valuable comments and discussions. We are also grateful to the referees for their detailed and helpful comments.  
}
 
\section*{\revise{Appendix:} List of symbols}
\begin{longtable}[H]{lp{9cm}p{5cm}l}
\textbf{Symbol} & \textbf{Meaning}                        & \textbf{Definition}                   & \textbf{Reference}                      \\ \cline{1-4}
\textbf{A} & & & \\
$\mathcal{A}_t$ & set of assets, excluding cash, at time $t$ & -- & p.\,\pageref{ref: A_t} \\
& & & \\
\textbf{B} & & & \\
$B_t$ & bank account at time $t$ & $B_t = \prod_{j=0}^{t-1}(1+F_j)$ & p.\,\pageref{ref: B_t} \\
$BE$ & best estimate & $
BE=E[\sum_{t=1}^TB_t^{-1}(gbf_t+gbf_t^{\le0}+ph_t + co_t - pr_t)]
%
$ & p.\,\pageref{ref: BE} \\
$BV_t$  & book value of the asset portfolio at time $t$ & $BV_t = \sum_{a\in\mathcal{A}_t}BV_t^a + C_t$ & p.\,\pageref{ref: BV_t} \\
$BV_t^a$ & book value of asset $a$ & -- & p.\,\pageref{ref: BV_t^a} \\
& & & \\
\textbf{C} & & & \\
$C_t$ & amount of cash held by the company at time $t$ & & p.\,\pageref{ref: C_t} \\
$cf_t^a$ & cash flow of asset $a$ at time $t$ & -- & p.\,\pageref{ref: cf_t^a} \\
$\chi_t$ & surrender fee factor at time $t$ & -- & p.\,\pageref{ref: chi_t} \\
$COG$ & cost of guarantees & $COG = E\left[\sum_{t=1}^T B_t^{-1} gs_t^-\right]$ & p.\,\pageref{ref: COG} \\
$co_t$ & cost cash flows at time $t$ & -- & p.\,\pageref{ref: co_t} \\
$CV_{s,t}^1$ & first coefficient of variation & $CV_{s,t}^1 = CV\biggl[D(s,t)-D(s,t+1)\biggr]$ & p.\,\pageref{ref: CV_s,t^1}  \\
$CV_s^2$ & second coefficient of variation & $CV_s^2 = CV\biggl[B_s^{-1}gs_s^+\biggr]$ & p.\,\pageref{ref: CV_s^2} \\
& & & \\
\textbf{D} & & & \\
$d$ & duration & -- & p.\,\pageref{ref: d}\\
$D(t,s)$ & discount factor from $s$ to $t<s$ & $D(t,s) = \prod_{j=t}^{s-1}(1+F_j)^{-1}$ & p.\,\pageref{ref: D(t,s)} \\
$DB_t$ & declared bonuses after valuation time & $DB_t = \sum_{x\in\mathcal{X}_t}DB_t^x$ & p.\,\pageref{ref: DB_t} \\
$DB_t^x$ & declared bonuses after valuation time of model point $x$ at time $t$ & -- & p.\,\pageref{ref: DB_t^x} \\
$DB_t^-$ & account of declared bonuses before bonus declaration at time $t$ & -- & p.\,\pageref{ref: DB_t^-} \\
$DB_t^{\le0}$ & declared bonuses up to and including valuation time & $DB_t^{\le0} = \sum_{x\in\mathcal{X}_t}(DB_t^{\le0})^x$ & p.\,\pageref{ref: DB_t^le0} \\
$(DB_t^{\le0})^x$ & declared bonuses up to and including valuation time of model point $x$ at time $t$ & -- & p.\,\pageref{ref: (DB_t^le0)^x} \\
$\Delta f_t$ & increment of $f_t$ & $\Delta f_t = f_t - f_{t-1}$ & p.\,\pageref{ref: Delta f} \\
& & & \\
\textbf{E} & & & \\
$\eta_t$ & fraction of declaration of $SF_{t-1}$ to $DB_t$ & -- & p.\,\pageref{ref: eta_k} \\
& & & \\
\textbf{F} & & & \\
$F_t$ & simple one year forward rate between $t$ and $t+1$ & -- & p.\,\pageref{ref: F_t} \\
$FC_t$ & free capital at time $t$ & $FC_t = BV_t - L_t$ & p.\,\pageref{ref: FC_t} \\
$FDB$ & value of future discretionary benefits & $FDB = E\left[\sum_{t=1}^T B_t^{-1}ph_t\right]$ & p.\,\pageref{ref: FDB} \\
& & & \\
\textbf{G} & & & \\
$\gamma_t$ & fraction of technical gains & $\gamma_t = (tg_t + \chi_t^{\le0}DB_{t-1}^{\le0} + \chi_t DB_{t-1})/LP_{t-1}$ & p.\,\pageref{ref: gamma_t} \\
$GB$ & value of guaranteed benefits & $GB = BE - FDB$ & p.\,\pageref{ref: GB} \\
$gbf_t$ & guaranteed benefits at time $t$ & $gbf_t = \sum_{x\in\mathcal{X}_t}gbf_t^x$ & p.\,\pageref{ref: gbf_t} \\
$gbf_t^x$ & guaranteed benefits generated by model point $x$ at time $t$ & -- & p.\,\pageref{ref: gbf_t^x} \\
$gbf_t^{\le0}$ & cash flows due to $DB_{t-1}^{\le0}$ & $gbf_t^{\le0} = \sum_{x\in\mathcal{X}_t}(gbf_t^{\le0})^x$ & p.\,\pageref{ref: gbf_t^le0} \\
$(gbf_t^{\le0})^x$  & cash flows due to $(DB_{t-1}^{\le0})^x$ & -- & p.\,\pageref{ref: (gbf_t^le0)^x} \\
$gph$ & policyholder share in gross surplus & -- & p.\,\pageref{ref: sh_t} \\
$gs_t$ & gross surplus at time $t$ & $gs_t = ROA_t - \Delta V_t - \Delta DB_t^{\le0} - DB_t^- + DB_{t-1} + pr_t - gbf_t - gbf_t^{\le0} - ph_t - co_t$ & p.\,\pageref{ref: gs_t} \\
$gsh$ & share holder share in gross surplus & -- & p.\,\pageref{ref: sh_t} \\
$gtax$ & tax paid on gross surplus at time $t$ & -- & p.\,\pageref{ref: tax_t} \\
& & & \\
\textbf{H} & & & \\
$h$ & half life of assurance provisions & -- & p.\,\pageref{ref: h} \\
& & & \\
\textbf{I} & & & \\
$IV_s$ & caplet implied volatility & -- & p.\,\pageref{ref: IV_s} \\
& & & \\
\textbf{K} & & & \\
$k_s$ & strike & -- & p.\,\pageref{ref: k_s} \\
& & & \\
\textbf{L} & & & \\
$L_t$           & book value of liabilities at time $t$                             & $L_t = LP_t + SF_t$                   & p.\,\pageref{ref: L_t} \\
$LP_t$ & life assurance provision at time $t$ & $LP_t = V_t + DB_t^{\le0} + DB_t$ & p.\,\pageref{ref: LP_t} \\
& & & \\
\textbf{M} & & & \\
$\mu_k^t$ & fraction of bonus declarations from time $k$ paid out (or kept as surrender fee) at $t$ & -- & p.\,\pageref{ref: mu_k^t} \\
$MV_t$ & market value of the portfolio at time $t$ & $MV_t = \sum_{a\in\mathcal{A}_t}MV_t^a + C_t$  & p.\,\pageref{ref: MV_t} \\
$MV_t^a$ & market value of asset $a$ at time $t$ & -- & p.\,\pageref{ref: MV_t^a} \\
& & & \\
\textbf{N} & & & \\
$\nu$ & bonus declaration bound & -- & p.\,\pageref{ref: nu} \\
$\nu_t$ & declaration fraction of $ph_t^*$ & -- & p.\,\pageref{ref: nu_t} \\
& & & \\
\textbf{O} & & & \\
$O_s^+$ & value of the caplet with maturity $s-1$ & -- & p.\,\pageref{ref: O_s} \\
$O_s^-$ & value of the floorlet with maturity $s-1$ & -- & p.\,\pageref{ref: O_s} \\
& & & \\
\textbf{P} & & & \\
$P(t,s)$ & value of a zero coupon bond, with nominal of $1$ at $s$, at time $t$ & $P(t,s) = E[D(t,s)]$ & p.\,\pageref{ref: P(t,s)} \\
$PH^*$ & time value of the accounting flows $ph_t^*$ & $PH^* = E\left[\sum_{t=1}^T B_t^{-1}ph_t^*\right]$ & p.\,\pageref{ref: PH^*} \\
$ph_t$ & amount of discretionary benefits paid out at time $t$ & $ph_t = \sum_{x\in\mathcal{X}_t}ph_t^x$ & p.\,\pageref{ref: ph_t} \\
$ph_t^*$ & policyholder accounting flow at time $t$ & $ph_t^* = gph\cdot gs_t^+$ & p.\,\pageref{ref: ph_t^*} \\
$ph_t^x$ & cash flows due to $DB_{t-1}^x$ & -- & p.\,\pageref{ref: ph_t^x} \\
$pr_t$ & premium payments at time $t$ & -- & p.\,\pageref{ref: pr_t} \\
& & & \\
\textbf{R} & & & \\
$\rho_t$ & average technical interest rate at time $t-1$ & -- & p.\,\pageref{ref: rho_t} \\
$ROA_t$ & book value return at time $t$ & $ROA_t = \sum_{a\in\mathcal{A}_{t-1}}ROA_t^a + F_{t-1}C_{t-1}$ & p.\,\pageref{ref: ROA_t} \\
$ROA_t^a$ & book value return of asset $a$ at time $t$ & $ROA_t^a = cf_t^a + \Delta BV_t^a$ & p.\,\pageref{ref: ROA_t^a} \\
& & & \\
\textbf{S} & & & \\
$SF_t$ & surplus fund at time $t$ & -- & p.\,\pageref{ref: SF_t} \\
$sg_t^*$ & surrender fee at time $t$ & -- & p.\,\pageref{ref: sg_t^*} \\
$sh_t$ & share holder cash flow at time $t$ & $sh_t = gsh\cdot gs_t^+ - gs_t^-$ & p.\,\pageref{ref: sh_t} \\
& & & \\
\textbf{T} & & & \\
$T$ & projection horizon & -- & p.\,\pageref{ref: T} \\
$TAX$ & time value of tax & $TAX = E\left[\sum_{t=1}^T B_t^{-1}tax_t\right]$ & p.\,\pageref{ref: TAX} \\
$tax_t$ & tax cash flow at time $t$ & $tax_t = gtax\cdot gs_t^+$ & p.\,\pageref{ref: tax_t} \\
$tg_t$ & technical gains at time $t$ & -- & p.\,\pageref{ref: tg_t} \\
$\theta$ & surplus fund fraction & -- & p.\,\pageref{ref: theta} \\
& & & \\
\textbf{U} & & & \\
$UG_t$ & unrealized gains at time $t$ & $UG_t = MV_t - BV_t$ & p.\,\pageref{ref: UG_t} \\
$UG_t^a$ & unrealized gains of asset $a$ at time $t$ & $UG_t^a = MV_t^a - BV_t^a$ & p.\,\pageref{ref: UG_t^a} \\
& & & \\
\textbf{V} & & & \\
$V_t$ & mathematical reserves at time $t$ & $V_t = \sum_{x\in\mathcal{X}_t}V_t^x$ & p.\,\pageref{ref: V_t} \\
$V_t^x$ & mathematical reserve of model point $x$ at time $t$ & -- & p.\,\pageref{ref: V_t^x} \\
$VIF$ & value of in-force business & $VIF = E\left[\sum_{t=1}^T B_t^{-1}sh_t\right]$ & p.\,\pageref{ref: VIF} \\
& & & \\
\textbf{X} & & & \\
$\mathcal{X}_t$ & set of model points active at time $t$ & -- & p.\,\pageref{ref: X_t} 
\end{longtable}

\end{document}